\documentclass[11pt,a4paper,leqno]{article}
\usepackage{graphicx}
\usepackage{makeidx}
\usepackage{pst-node}
\usepackage{amsmath,amsthm}
\usepackage{enumerate}
\usepackage{verbatim}
\usepackage{natbib}

\newtheorem{theorem}{Theorem}
\newtheorem{lemma}[theorem]{Lemma} 
\newtheorem{proposition}[theorem]{Proposition}

\theoremstyle{definition}
\newtheorem{remark}{Remark}

\title{Complexity of inheritance of $\mathcal{F}$-convexity
for restricted games induced by minimum partitions.}

\author{
A. Skoda\thanks{Corresponding
    author. Universit\'e de Paris I, Centre d'Economie de la Sorbonne, 106-112
    Bd de l'H\^opital, 75013 Paris, France. E-mail: {\tt
      alexandre.skoda@univ-paris1.fr}}
}

\begin{document}
\def\R{{\bf \mbox{I\hspace{-.17em}R}}}
\def\N{\mbox{I\hspace{-.15em}N}}

\maketitle

\begin{abstract}
Let $G=(N,E,w)$ be a weighted communication graph
(with weight function $w$ on $E$).
For every subset $A \subseteq N$,
we delete in the subset $E(A)$
of edges with ends in $A$,
all edges of minimum weight in $E(A)$.
Then the connected components of the corresponding induced subgraph
constitute a partition of $A$
that we call $\mathcal{P}_{\min}(A)$.
For every game $(N,v)$,
we define the $\mathcal{P}_{\min}$-restricted game $(N, \overline{v})$
by $\overline{v}(A) = \sum_{F \in \mathcal{P}_{\min}(A)} v(F)$
for all $A \subseteq N$.
We prove that we can decide in polynomial time
if there is inheritance of $\mathcal{F}$-convexity
from $(N,v)$ to the $\mathcal{P}_{\min}$-restricted game $(N, \overline{v})$
where $\mathcal{F}$-convexity is obtained
by restricting convexity to connected subsets.
\end{abstract}



\noindent

\section{Introduction}
We consider, 
on a given finite set $N$ with $|N| = n$, a weighted communication structure,
\emph{i.e.},
a weighted gragh $G=(N, E, w)$
where $w$ is a weight function defined on the set $E$ of edges of $G$.
For a given subset $A$ of $N$,
we denote by $E(A)$ the set of edges in $E$ with both endvertices in $A$,
and by $\Sigma(A)$ the subset of edges of minimum weight in $E(A)$.
Let $G_{A}$ be the graph induced by $A$,
\emph{i.e.},
$G_{A} = (A, E(A))$.
Then $\tilde{G}_{A} = (A, E(A) \setminus \Sigma(A))$
is the graph obtained by deleting the minimum weight edges in $G_{A}$.
In \citep{Skoda2016}
we introduced 
the correspondence $\mathcal{P}_{\min}$ on $N$ which associates to every subset $A \subseteq N$, 
the partition $\mathcal{P}_{\min}(A)$ of $A$ into the connected components of $\tilde{G}_A$.
Then for every game $(N,v)$
we defined the restricted game $(N,\overline{v})$
associated with $\mathcal{P}_{\min}$ by:
\begin{equation}
\overline {v} (A) = \sum_{F \in \mathcal{P}_{\min}(A)} v(F),  \; \; \textrm{for all}\; A \subseteq N.
\end{equation}
We more simply refer to this game as the $\mathcal{P}_{\min}$-restricted game.
$v$ is the characteristic function of the game,
$v: 2^{N} \rightarrow \R$, $A \mapsto v(A)$
and satisfies $v(\emptyset) = 0$.
Compared to the initial game $(N,v)$,
the restricted game $(N, \overline {v})$ takes into account
the combinatorial structure of the graph
and the ability of players to cooperate in a given coalition
and therefore differents aspects of cooperation restrictions.
In particular,
assuming that the edge weights reflect the degrees of relationships between players,
$\mathcal{P}_{\min}(A)$ gives a partition of a coalition $A$
into subgroups where players are in privileged relationships
(with respect to the minimum relationship degree in $A$).
Many other correspondences have been considered
to define restricted games
(see, e.g.,
\citet{Myerson77,AlgabaBilbaoLopez2001,Bilbao2000,Bilbao2003,Faigle89,GrabischSkoda2012,Grabisch2013}).
For a given correspondence
a classical problem is to study the inheritance of basic properties
as superadditivity and convexity
from the underlying game to the restricted game.
Inheritance of convexity is of particular interest
as it implies that good properties are inherited,
for instance the non-emptyness of the core,
and that the Shapley value is in the core.
For the $\mathcal{P}_{\min}$ correspondence
we proved in \citep{GrabischSkoda2012} that
we always have inheritance of superadditivity
from $(N,v)$ to $(N, \overline{v})$.
Let us observe
that inheritance of convexity is a strong property.
Hence it would be useful to consider weaker properties than convexity.
Following alternative definitions of convexity
in combinatorial optimization
and game theory
when restricted families of subsets are considered
(not necessarily closed under union and intersection),
see, \emph{e.g.},
\citep{EdmondsGiles77,Faigle89,Fujishige2005}
we introduced in \citep{GrabischSkoda2012}
the $\mathcal{F}$-convexity by restricting convexity to the family $\mathcal{F}$
of connected subsets of $G$.
In \citep{Skoda2016},
we have characterized inheritance of $\mathcal{F}$-convexity for $\mathcal{P}_{\min}$
by five necessary and sufficient conditions on the edge-weights.
Of course the study of inheritance of $\mathcal{F}$-convexity
is also a first key step to characterize
inheritance of convexity in the general case.
We have also highlighted in \citep{Skoda2016}
a relation between Myerson's restricted game
introduced in \citep{Myerson77} and the $\mathcal{P}_{\min}$-restricted game.
Myerson's restricted game $(N, v^{M})$ is defined by
$v^{M} (A) = \sum_{F \in \mathcal{P}_{M}(A)} v(F)$
for all $A \subseteq N$,
where $\mathcal{P}_M (A)$ is the set of connected components of $G_{A}$.
We proved that inheritance of convexity for the Myerson's game is equivalent
to inheritance of $\mathcal{F}$-convexity
for the $\mathcal{P}_{\min}$-restricted game
associated with a weighted graph with only two different weights.
In the present paper
we prove that inheritance
of $\mathcal{F}$-convexity
can be checked in polynomial time for $\mathcal{P}_{\min}$.
Of course,
directly testing inheritance of convexity
(resp. $\mathcal{F}$-convexity)
for the $\mathcal{P}_{\min}$-correspondence
is highly non polynomial.
We have to consider all convex games $(N,v)$
and for each of them, we have to compute the $\mathcal{P}_{\min}$-restricted game $(N,\overline{v})$
and to verify the convexity (resp. $\mathcal{F}$-convexity) of $(N,\overline{v})$,
\emph{i.e.},
for all subsets $A, B\in 2^N$
(resp. for all subsets $A$, $B$ such that $A$, $B$, and $A\cap B$ are connected),
to verify that $v(A\cup B)+v(A\cap B)\ge v(A)+v(B)$.
The characterization given in \citep{Skoda2016}
implies that we can verify inheritance of $\mathcal{F}$-convexity
by checking 
that some specific
distributions of edge-weights are satisfied on all stars, paths, cycles,
pans\footnote{A pan graph is a graph obtained by joining a cycle to a vertex by an edge.},
and adjacent cycles of the graph $G$.
Of course directly checking all these conditions would also lead to a non-polynomial algorithm
as the number of paths and cycles can be exponential.
We prove that we only have to take into account
a polynomial number of specific paths and cycles
and therefore
we can decide in $O(n^6)$ time
whether there is
inheritance of $\mathcal{F}$-convexity.

Moreover,
to establish this result
we have to more deeply analyze the relations between the
star, path, cycle, pan and adjacent cycles conditions
characterizing inheritance of the $\mathcal{F}$-convexity.
It is of independent interest
as we prove that
some of them are not completely independent of the others 
and we highlight their relations 
which for the sake of simplicity
were only implicitely contained in \citep{Skoda2016}.
In particular we observe that
the cycle condition in \citep{Skoda2016} splits into
an intermediary cycle condition
which is a consequence of the star and path conditions
and a condition on chords of the cycles.
Also a part of the pan condition is a consequence of star and path conditions.
We prove that
the cycle condition is a consequence of the  star, path and adjacent cycles conditions.
Hence inheritance of $\mathcal{F}$-convexity
can be characterized by only
four of the previous conditions:
the star, path, pan, and adjacent cycles conditions.

Star condition can easily be checked in polynomial time.
Checking the other conditions in polynomial time requires a deeper and careful analysis.
In a first fundamental step,
we prove that
to verify path condition in polynomial time
we only have to consider one minimum weight spanning tree $T$ in $G$
and to verify path condition for the paths
joining the end vertices
of $T$
and to test intermediary cycle condition for the fundamental cycles of $T$.
Then the main step, which is also the most difficult one,
deals with the adjacent cycles condition
which
\textit{a priori} involves any pair of adjacent cycles $(C, C')$.
A first contradiction to the adjacent cycle condition
is the existence of two cycles $C$ anc $C^{'}$
with two common adjacent edges $e_{1}$ and $e_{2}$ of non maximum weight.
We prove that
if there exist two such cycles
then there also exist two specific cycles $\tilde{C}$ and  $\tilde{C}^{'}$
with two common edges of non maximum weight
which are
associated with shortest paths in some particular subgraphs 
depending only on the pair $(e_{1}, e_{2})$.
As a result of their features
such cycles $\tilde{C}$ and  $\tilde{C}^{'}$
can be detected in polynomial time.
A similar result holds for the second possible contradiction
to the adjacent cycle condition
which corresponds to the existence of two  adjacent cycles $C$ and $C^{'}$ 
with only one common non maximum weight edge $e_{1}$.
Hence, looking for $\tilde{C}$ and $\tilde{C}^{'}$
for any pair of adjacent edges (resp. for any edge) in $G$,
we can detect any contradiction to the adjacent cycles condition.
Therefore the adjacent cycles condition can be checked in polynomial time.
Assuming that the star, path and adjacent cycles condition
are satisfied,
the cycle condition is necessarily satisfied.
Finally using the cycle condition
we establish a new condition equivalent to the pan condition
that we can easily verify in polynomial time
using an appropriate shortest path.\\
The article is organized as follows.
In Section~\ref{SectionPreliminaryDefinitions},
we give preliminary definitions and results.
In particular,
we recall the definitions of convexity
and
$\mathcal{F}$-convexity.
In section~\ref{subsubsectionNecessaryCondF-ConvPmin},
we recall necessary and sufficient conditions
on graph edge-weights
to have inheritance of $\mathcal{F}$-convexity
for the correspondence $\mathcal{P}_{\min}$ established in \citep{Skoda2016}.
Then section~\ref{SectionComplexityOfInheritance} is the main part of this article
where
we establish that inheritance of $\mathcal{F}$-convexity for the $\mathcal{P}_{\min}$-correspondence
can be checked in polynomial time.

\section{Preliminary definitions and results}
\label{SectionPreliminaryDefinitions}

Let $N$ be a given set.
We denote by $2^{N}$ the set of all subsets of $N$.
A game $(N,v)$ is
\textit{zero-normalized}\label{definitionZeroNormalizedGame}
if $v(i) = 0$ for all $i \in N$.
A game $(N,v)$ is 
\textit{superadditive}
if, for all $A,B \in 2^{N}$ such that $A \cap B = \emptyset$, $v(A \cup B) \geq v(A) + v(B)$.
For any given subset $\emptyset \not= S \subseteq N$, the unanimity game $(N, u_{S})$ is defined by:
\begin{equation}
u_{S}(A) =
\left \lbrace
\begin{array}{cl}
1 &  \textrm{ if } A \supseteq S,\\
0 & \textrm{ otherwise.}
\end{array}
\right.
\end{equation}
We note that $u_{S}$ is superadditive for all $S \not= \emptyset$.

Let us consider a game $(N,v)$.
For arbitrary subsets $A$ and $B$ of $N$, we define the value:
\[
\Delta v(A,B):=v(A\cup B)+v(A\cap B)-v(A)-v(B).
\]
A game $(N,v)$ is \textit{convex}
if its characteristic function $v$ is supermodular,
\emph{i.e.},
$\Delta v(A,B)\geq 0$
for all $A,B \in 2^{N}$.
We note that $u_{S}$ is supermodular for all $S \not= \emptyset$.
Let $\mathcal{F}$
be a
\emph{weakly union-closed family}
\footnote{
$\mathcal{F}$ 
is weakly union-closed if $A \cup B \in \mathcal{F}$ for all $A$, $B \in \mathcal{F}$
such that $A \cap B \not= \emptyset $ \citep{FaigleGrabisch2010}.
Weakly union-closed families were introduced
and analysed in \cite{Algaba98,AlgabaBilbaoBormLopez2000} and called union stable systems.
}
of subsets of $N$
such that $\emptyset \notin \mathcal{F}$.
A game $v$ on $2^{N}$ is said to be \emph{$\mathcal{F}$-convex}
if $\Delta v(A,B) \geq 0$,
for all $A,B \in \mathcal{F}$
such that $A \cap B \in \mathcal{F}$.
If $\mathcal{F} = 2^{N} \setminus \lbrace \emptyset \rbrace$
then $\mathcal{F}$-convexity corresponds to
convexity.

For a given graph $G = (N, E)$,
we say that a subset
$A \subseteq N$ is connected
if the induced graph $G_{A} = (A,E(A))$
is connected.
In this paper
$\mathcal{F}$ will be the family of connected subsets of $N$.

We recall the following theorems established in \citep{Skoda2016}.

\begin{theorem}
\label{thLGNEPNFaSNuSsFaSNuSin}
Let $G=(N,E,w)$ be an arbitrary weighted graph,
and $\mathcal{P}$ an arbitrary correspondence on $N$.
The following claims are equivalent:
\begin{enumerate}[1)]
\item
For all $\emptyset \not= S \subseteq N$, $\overline{u_{S}}$ is superadditive.
\item
\label{Claim3ThLGNEPNFaSNuSsFaSNuSin}
For all subsets $A \subseteq B \subseteq N$,
$\mathcal{P}(A)$ is a refinement of the restriction of $\mathcal{P}(B)$ to $A$.
\item
For all superadditive game $(N,v)$
the $\mathcal{P}$-restricted game $(N,\overline{v})$ is superadditive.
\end{enumerate}
\end{theorem}

\begin{theorem}
\label{thLGNuPNFNFSNuSSNNuSFABFABF}
Let $G=(N,E,w)$ be an arbitrary weighted graph,
and
$\mathcal{P}$ an arbitrary correspondence on $N$.
If for each non-empty subset $S \subseteq N$, $\overline{u_{S}}$ is superadditive,
then the following claims are equivalent.
\begin{enumerate}[1)]
\item
For all non-empty subset $S \subseteq N $,
the game $(N, \overline{u_{S}})$ is $\mathcal{F}$-convex.
\item
For all $i \in N$ and for all $A, B \in \mathcal{F}$,
$ A \subseteq B \subseteq N \setminus \lbrace i \rbrace$
such that $A \cup \lbrace i \rbrace  \in \mathcal{F}$,
we have for all $A^{'} \in \mathcal{P}(A \cup \lbrace i \rbrace)$,
$\mathcal{P}(A)_{|A^{'}} = \mathcal{P}(B)_{|A^{'}}$.
\end{enumerate}
\end{theorem}

Let $G = (N,E,w)$ be a weighted graph.
We set $|N| = n$ and $|E| = m$.
We assume that all weights are strictly positive
and denote by $w_{k}$ or $w_{ij}$
the weight of an edge $e_{k} = \lbrace i,j \rbrace$ in $E$.
The adjacency matrix $A$ of $G$
is an $n \times n$ matrix defined by:
\[
a_{ij} =
\left \lbrace
\begin{array}{cl}
w_{ij} & \textrm{if } \lbrace i,j \rbrace \in E, \\
0 & \textrm{otherwise}.
\end{array}
\right. 
\]

In the next section we recall
necessary and sufficient conditions on the weight vector $w$
established in \citep{Skoda2016} 
for the inheritance of $\mathcal{F}$-convexity
from the original communication game $(N,v)$ 
to the $\mathcal{P}_{\min}$-restricted game $(N,\overline{v})$.

\section{Necessary
conditions on edge-weights}
\label{subsubsectionNecessaryCondF-ConvPmin}

A star $S_k$ corresponds to a tree
with one internal vertex and $k$ leaves.
We consider a star $S_{3}$
with vertices ${1, 2, 3, 4}$ and edges $e_{1} = \lbrace 1, 2 \rbrace$,
$e_{2} = \lbrace 1, 3 \rbrace$ and $e_{3} = \lbrace 1, 4 \rbrace$.\\

\noindent
\framebox[\linewidth]{
\parbox{.95\linewidth}{
\textbf{Star Condition.}
\it
For every star of type $S_{3}$ of $G$, the edge-weights $w_{1}, w_{2}, w_{3}$ satisfy, after renumbering the edges if necessary:
\[
 w_{1} \leq w_{2} = w_{3}.
\]
}}

\begin{proposition}\label{propStarCondInheritanceF-convexityP_min}
Let $G=(N,E,w)$ be a weighted graph.
If for every una\-ni\-mi\-ty game $(N,u_{S})$,
the $\mathcal{P}_{\min}$-restricted game $(N,\overline{u_{S}})$ is $\mathcal{F}$-convex,
then the Star Condition is satisfied.
\end{proposition}

\noindent 
\framebox[\linewidth]{
\parbox{.95\linewidth}{
\textbf{Path Condition.}
\it
For every elementary path $\gamma = (1, e_{1}, 2, e_{2}, 3, \ldots,  m, \\ e_{m}, m+1)$ in $G$
and for all $i,j,k$
such that $1 \leq i < j < k \leq m$,
the edge-weights satisfy:
\[
w_{j} \leq \max (w_{i}, w_{k}).
\]
}}

\begin{proposition}\label{corPathCond} 
Let $G = (N,E,w)$ be a weighted graph.
If for every unanimity game $(N, u_{S})$,
the
$\mathcal{P}_{\min}$-
restricted game $(N, \overline{u_{S}})$ is $\mathcal{F}$-convex,
then the Path Condition is satisfied.
\end{proposition}

For a given cycle in $G$,
$C = \lbrace 1, e_{1}, 2, e_{2}, \ldots, m, e_{m}, 1\rbrace$ with $m \geq 3$,
we denote by $E(C)$ the set of edges $\lbrace e_{1}, e_{2}, \ldots, e_{m}\rbrace$ of $C$
and by $\hat{E}(C)$ the set composed of $E(C)$ and of the chords of $C$ in $G$.\\

\noindent 
\framebox[\linewidth]{
\parbox{.95\linewidth}{
\textbf{Weak Cycle Condition.}
\it
For every simple cycle of $G$,
$C = \lbrace 1, e_{1}, 2, e_{2}, \ldots, m, e_{m}, 1\rbrace$ with $m \geq 3$,
the edge-weights satisfy, after renumbering the edges if necessary:
\begin{equation}\label{eqCoru1<=u2<=u3=um=M}
w_{1} \leq w_{2} \leq w_{3} = \cdots = w_{m} = M
\end{equation}
where $M = \max_{e \in E(C)} w(e)$.
}}

\noindent 
\framebox[\linewidth]{
\parbox{.95\linewidth}{
\textbf{Intermediary Cycle Condition.}
\it
\begin{enumerate}[1)]
\item
Weak Cycle condition.
\item
Moreover
$w(e) \leq w_{2}$ for all chord incident to $2$,
and $w(e) \leq \hat{M} = \max_{e \in \hat{E}(C)}$ for all chord non incident to $2$.
Moreover:
\begin{itemize}
\item
If $w_{1} \leq w_{2} < \hat{M}$
then $w(e) = \hat{M}$ for all $e \in \hat{E}(C)$  non-incident to~$2$.
If $e$ is a chord incident to $2$
then $w_{1} \leq w_{2} = w(e) < \hat{M}$
or $w(e) < w_{1} = w_{2} < \hat{M}$.
\item
If $w_{1} < w_{2} = \hat{M}$,
then $w(e) = \hat{M}$ for all $e \in \hat{E}(C) \setminus \lbrace e_{1} \rbrace$.
\end{itemize}
\end{enumerate}
}}

\noindent 
\framebox[\linewidth]{
\parbox{.95\linewidth}{
\textbf{Cycle Condition.}
\it
\begin{enumerate}[1)]
\item
Weak Cycle condition.
\item
Moreover
$w(e) = w_{2}$ for all chord incident to $2$,
and $w(e) = \hat{M}$ for all $e \in \hat{E}(C)$ non incident to $2$.
\end{enumerate}
}}

\begin{proposition}\label{CorCycleConditionGeneralPmin}
Let $G = (N,E,w)$ be a weighted graph.
\begin{enumerate}
\item
\label{itemPathConditionImpliesAWeakCycleCondition}
If $G$ satisfies the Path condition
then the Weak Cycle condition is satisfied.
\item
\label{itemPathAndStarConditionsImplyAnIntermediaryCycleCondition}
If $G$ satisfies the Star and Path conditions
then the Intermediary Cycle condition is satisfied.
\item
\label{itemCycleConditionGeneralPmin}
If for every unanimity game $(N, u_{S})$,
the
$\mathcal{P}_{\min}$-
restricted game $(N, \overline{u_{S}})$ is $\mathcal{F}$-convex,
then the Cycle Condition is satisfied.
\end{enumerate}
\end{proposition}

\noindent 
\framebox[\linewidth]{
\parbox{.95\linewidth}{
\textbf{Weak Pan Condition.}
\it
For all connected subgraphs corresponding to the union of a simple cycle $C = \lbrace e_{1}, e_{2}, \ldots, e_{m} \rbrace$
with $m \geq 3$,
and an elementary path $P$ such that there is an edge $e$ in $P$
with $w(e) \leq \min_{1 \leq k \leq m} w_{k}$
and $|V(C) \cap V(P)| = 1$,
the edge-weights satisfy:
\begin{enumerate}[(a)]
\item
\label{eqPanConditionEither}
either $w_{1} = w_{2} = w_{3} = \cdots = w_{m}= \hat{M}$,
\item
\label{eqPanConditionOr}
or $w_{1} = w_{2} < w_{3} = \cdots = w_{m} = \hat{M}$,
\end{enumerate}
where $\hat{M} = \max_{e \in \hat{E}(C)} w(e)$.
In this last case $V(C) \cap V(P) = \lbrace 2 \rbrace$.
}}

\noindent 
\framebox[\linewidth]{
\parbox{.95\linewidth}{
\textbf{Pan Condition.}
\it
\begin{enumerate}[1)]
\item
Weak Pan condition.
\item
If Claim (\ref{eqPanConditionOr}) of the Weak Pan condition is satisfied and if moreover $w(e) < w_{1}$
then $\lbrace 1, 3 \rbrace$ is a maximum weight chord of $C$.
\end{enumerate}
}}

\begin{proposition}
\label{PropPanCondition}
Let $G = (N,E,w)$ be a weighted graph.
\begin{enumerate}[1)]
\item
\label{ItemPropPanConditionIfGSatisfiesStarAndPath}
If $G$ satisfies the Star, and Path conditions
then the Weak Pan condition is satisfied.
\item
\label{ItemPropPanConditionIfForEveryUnanimity}
If for every unanimity game $(N, u_{S})$,
the $\mathcal{P}_{\min}$-restricted game $(N, \overline{u_{S}})$ is $\mathcal{F}$-convex,
then the Pan Condition is satisfied.
\end{enumerate}
\end{proposition}

We recall the proof of Proposition~\ref{PropPanCondition}
as we will need it to prove a new formulation of the Pan condition
in Section~\ref{SectionComplexityOfInheritance}.

\begin{proof}
\ref{ItemPropPanConditionIfGSatisfiesStarAndPath})
Let us consider $C = \lbrace 1, e_{1}, 2, e_{2}, 3, \ldots, m, e_{m}, 1 \rbrace$,
and $P = \lbrace j, e_{m+1},\\ m+1, e_{m+2}, m+2, \ldots, e_{m+r}, m+r \rbrace$
with $j \in \lbrace 1, \ldots, m  \rbrace$,
as represented in Figure~\ref{figPanConditionCycle+Path}.
We can assume w.l.o.g. that $e = e_{m+r}$
(restricting $P$ if necessary)
and that $w_{m+j} > w_{m+r} = w(e)$
for all $1 \leq j \leq r-1$.
\begin{figure}[!h]
\centering
\begin{pspicture}(0,-.3)(0,2.5)
\tiny
\begin{psmatrix}[mnode=circle,colsep=0.4,rowsep=0.5]
	& { }	&  	&  { } \\
{$m$} 	& 	&	& 	&  {$j$}  & &  { } & & { } & & {}\\
	& {$1$} 	& 	& {$2$}
\psset{arrows=-, shortput=nab,labelsep={0.1}}
\tiny
\ncline{3,2}{3,4}_{$e_{1}$}
\ncline{2,1}{3,2}_{$e_{m}$}
\ncline{2,1}{1,2}
\ncline{3,4}{2,5}_{$e_{2}$}
\ncline{2,5}{1,4}_{$e_{j}$}
\ncline{1,2}{1,4}
\ncline{2,5}{2,7}_{$e_{m+1}$}
\ncline{2,7}{2,9}_{$e_{m+2}$}
\ncline{2,9}{2,11}_{$e_{m+r}$}
\normalsize
\end{psmatrix}
\end{pspicture}
\caption{Pan formed by the union of $C$ and $P$.}
\label{figPanConditionCycle+Path}
\end{figure}
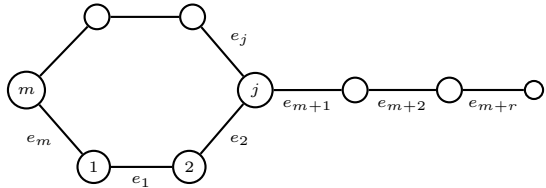
Applying Claim~\ref{itemPathAndStarConditionsImplyAnIntermediaryCycleCondition}
of Proposition~\ref{CorCycleConditionGeneralPmin}
to the cycle $C$,
we have after renumbering the edges if necessary:
\begin{equation}\label{eqproofu1<=u2<=u3=um=M}
w_{1} \leq w_{2} \leq w_{3} = \cdots = w_{m} = \hat{M}. 
\end{equation}
Let us first assume $3 \leq j \leq m$.
Applying Proposition~\ref{corPathCond} to the path
$\lbrace 2, e_{1}, 1, e_{m}, m, \ldots, j+1, e_{j}, j, e_{m+1}, m+1,
\ldots, m+r-1, e_{m+r}, m+r \rbrace$,
we have $w_{j} \leq \max (w_{1}, w(e)) = w_{1}$.
Then (\ref{eqproofu1<=u2<=u3=um=M}) implies (\ref{eqPanConditionEither}).\\
Let us now assume $j \in \lbrace 1, 2 \rbrace$.
If $r=1$
then $w_{m+1} = w(e) \leq w_{1}$.
Otherwise,
applying Proposition~\ref{corPathCond} to the path
$\lbrace  e_{1}, e_{m+1}\ldots, e_{m+r} \rbrace$,
we have $w_{m+1} \leq \max (w_{1}, w(e)) = w_{1}$.\\
If $j=1$,
Proposition~\ref{propStarCondInheritanceF-convexityP_min}
applied to the star defined by $\lbrace e_{1}, e_{m}, e_{m+1} \rbrace$,
implies $w_{m+1} \leq w_{1} = w_{m}$.
Hence (\ref{eqproofu1<=u2<=u3=um=M}) still implies (\ref{eqPanConditionEither}).\\
If $j=2$,
Proposition~\ref{propStarCondInheritanceF-convexityP_min}
applied to the star defined by $\lbrace e_{1}, e_{2}, e_{m+1} \rbrace$,
implies $w_{m+1} \leq w_{1} = w_{2}$.
If $w_{1} = w_{2} = \hat{M}$
then (\ref{eqPanConditionEither}) is satisfied.
Otherwise we have $w_{1} = w_{2} < \hat{M}$
and (\ref{eqPanConditionOr}) is satisfied.
Hence the Weak Pan condition is satisfied.\\
\ref{ItemPropPanConditionIfForEveryUnanimity})
Now if Claim (\ref{eqPanConditionOr}) of the Weak Pan condition is satisfied,
let us assume by contradiction that $\lbrace 1, 3 \rbrace \notin E(C)$.
Claim~\ref{itemCycleConditionGeneralPmin}
of Proposition~\ref{CorCycleConditionGeneralPmin} implies
$w(e) = \hat{M}$ (resp. $w(e) = w_{2}$) for any chord $e$ of $C$
non incident (resp. incident) to $2$.
Therefore we can assume w.l.o.g.
that $C$ has no maximum weight chord
(otherwise
we can replace $C$ by a smaller cycle which still contains the vertices $1, 2, 3$).
Let us consider $i \in V(C) \setminus \lbrace 1, 2, 3 \rbrace$,
$A = V(C) \setminus \lbrace i \rbrace$
and $B = A \cup V(P)$
as represented in Figure~\ref{figPanConditionCycle+Path2}.
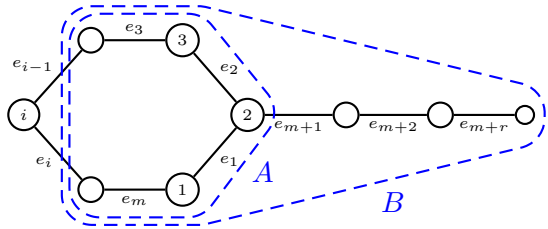
\begin{figure}[!h]
\centering
\begin{pspicture}(0,-.3)(0,2.5)
\tiny
\begin{psmatrix}[mnode=circle,colsep=0.4,rowsep=0.5]
	& { }	&  	&  {$3$} \\
{$i$} 	& 	&	& 	&  {$2$}  & &  { } & & { } & & {}\\
	& { } 	& 	& {$1$}
\psset{arrows=-, shortput=nab,labelsep={0.05}}
\tiny
\ncline{3,2}{3,4}_{$e_{m}$}
\ncline{2,1}{3,2}_{$e_{i}$}
\ncline{2,1}{1,2}^{$e_{i-1}$}
\ncline{3,4}{2,5}_{$e_{1}$}
\ncline{2,5}{1,4}_{$e_{2}$}
\ncline{1,2}{1,4}^{$e_{3}$}
\ncline{2,5}{2,7}_{$e_{m+1}$}
\ncline{2,7}{2,9}_{$e_{m+2}$}
\ncline{2,9}{2,11}_{$e_{m+r}$}
\normalsize
\uput[0](.8,.3){\textcolor{blue}{$A$}}
\pspolygon[framearc=1,linestyle=dashed,linecolor=blue,linearc=.4](-1.55,-.3)(-1.55,2.4)(.2,2.4)(1.25,1.1)(.2,-.3)
\uput[0](2.5,-.1){\textcolor{blue}{$B$}}
\pspolygon[framearc=1,linestyle=dashed,linecolor=blue,linearc=.4]
(-1.65,-.4)(-1.65,2.5)(.2,2.5)(4.7,1.4)(4.7,.7)(.2,-.4)
\end{psmatrix}
\end{pspicture}
\caption{$w_{m+r} < w_{1} = w_{2} < w_{3} = \cdots = w_{m} = M$.}
\label{figPanConditionCycle+Path2}
\end{figure}
Then $\mathcal{P}_{\min}(A) =
\lbrace \lbrace 2 \rbrace, \lbrace 3, 4, \ldots, i-1 \rbrace,
\lbrace i+1, \ldots, m, 1 \rbrace \rbrace$,
$\mathcal{P}_{\min}(A \cup \lbrace i \rbrace) =
\lbrace A \cup \lbrace i \rbrace \setminus \lbrace 2 \rbrace, \lbrace 2 \rbrace \rbrace$,
and
$\mathcal{P}_{\min}(B) = \lbrace B \setminus \lbrace m+r \rbrace, \lbrace m+r \rbrace \rbrace$
or $\mathcal{P}_{\min}(B) = \lbrace B \rbrace$
(this last case can occur if there exists an edge $e^{'}$ in $G_{B}$ with $w(e^{'}) <w(e)$
or with $m+r$ as an end-vertex and $w(e^{'}) > w(e)$).
Therefore $A^{'} :=  A \cup \lbrace i \rbrace \setminus \lbrace 2 \rbrace
\in \mathcal{P}_{\min}(A \cup \lbrace i \rbrace)$,
but
$\mathcal{P}_{\min}(B)_{|A^{'}} = 
\lbrace A \setminus \lbrace 2 \rbrace \rbrace
\not= \mathcal{P}_{\min}(A)_{|A{'}}$
and it contradicts Theorem~\ref{thLGNuPNFNFSNuSSNNuSFABFABF}.
\end{proof}

\begin{lemma}
\label{lemM=maxu(e)=maxu(e)=M'}
Let $G= (N,E,w)$ be a weighted graph satisfying the
Star and Path conditions.
Then for all pairs $(C,C^{'})$ of adjacent simple cycles in $G$,
we have:
\begin{equation}
\hat{M} =
\max_{e \in \hat{E}(C)} w(e) =
\max_{e \in E(C)} w(e) =
\max_{e \in E(C^{'})} w(e) =
\max_{e \in \hat{E}(C^{'})} w(e) = \hat{M}^{'}.
\end{equation}
\end{lemma}

\begin{proof}
We first consider $M = \max_{e \in E(C)} w(e)$ and $M^{'} = \max_{e \in E(C^{'})} w(e)$.
Let us consider two adjacent cycles $C$ and $C^{'}$
with $M < M^{'}$.
There is at least one edge $e_{1}$ common to $C$ and $C^{'}$.
Then we have $w_{1} \leq M < M^{'}$
and therefore $e_{1}$ is a non-maximum weight edge in $C^{'}$.
By Claim~\ref{itemPathConditionImpliesAWeakCycleCondition}
of Proposition~\ref{CorCycleConditionGeneralPmin}
the Weak Cycle condition is satisfied.
It implies that
there are at most two non-maximum weight edges in $C^{'}$.
Therefore there exists an edge $e_{2}^{'}$ in $C^{'}$ adjacent to $e_{1}$
with $w_{2}^{'} = M^{'}$.
As $M^{'} > M$,
$e_{2}^{'}$ is not an edge of $C$.
Let $e_{2}$ be the edge of $C$ adjacent to $e_{1}$ and $e_{2}^{'}$.
Then we have $w_{2} \leq M < M^{'}$
but it contradicts the Star condition applied to $\lbrace e_{1}, e_{2}, e_{2}^{'} \rbrace$
(two edge-weights are strictly smaller than $w_{2}^{'}$).
Therefore $M = M^{'}$.
Finally by Claim~\ref{itemPathAndStarConditionsImplyAnIntermediaryCycleCondition}
of Proposition~\ref{CorCycleConditionGeneralPmin}
the Intermediary Cycle condition is satisfied
and we have $\hat{M} = M = M^{'} = \hat{M}^{'}$.
\end{proof}

\noindent 
\framebox[\linewidth]{
\parbox{.95\linewidth}{
\textbf{Adjacent Cycles Condition.}
\it
For all pairs $(C,C^{'})$ of adjacent simple cycles in $G$ such that:
\begin{enumerate}[(a)]
\item
\label{enumPropV(C)-V(C)^{'}nonempty}
$V(C) \setminus V(C^{'}) \not= \emptyset$ and $V(C^{'}) \setminus V(C) \not= \emptyset$,
\item
\label{enumPropCatmost1non-maxweightchord}
$C$ has at most one non-maximum weight chord,
\item
\label{enumPropCC^{'}nomaxweightchord}
$C$ and $C^{'}$ have no maximum weight chord,
\item
\label{enumPropCandC'HaveNoCommonChord}
$C$ and $C^{'}$ have no common chord,
\end{enumerate}
\noindent
then $C$ and $C^{'}$ cannot have two common non-maximum weight edges.
Moreover $C$ and $C^{'}$ have a unique common non-maximum weight edge $e_{1}$ if and only if 
there are  non-maximum weight edges $e_{2} \in E(C) \setminus E(C^{'})$ and $e_{2}^{'} \in E(C^{'}) \setminus E(C)$
such that $e_{1}, e_{2}, e_{2}^{'}$ are adjacent and:
\begin{itemize}
\item $w_{1} = w_{2} = w_{2}^{'}$ if $|E(C)| \geq 4$ and $|E(C^{'})| \geq 4$.
\item $w_{1} = w_{2} \geq w_{2}^{'}$ or $w_{1} = w_{2}^{'} \geq w_{2}$ if $|E(C)| = 3$ or $|E(C^{'})| = 3$.
\end{itemize}
}}

\begin{proposition}\label{propAdjcyclesCond}
Let $G= (N,E,w)$ be a weighted graph.
If for every unanimity game $(N, u_{S})$
the $\mathcal{P}_{\min}$-restricted game $(N, \overline{u_{S}})$ is $\mathcal{F}$-convex,
then the Adjacent Cycles Condition is satisfied.
\end{proposition}


Finally the following characterization of inheritance of $\mathcal{F}$-convexity
was established  in  \citep{Skoda2016}.

\begin{theorem}\label{thPathBranchCyclePanAdjacentCond}
Let $G = (N,E,w)$ be a weighted graph.
For every superadditive and $\mathcal{F}$-convex game $(N,v)$,
the $\mathcal{P}_{\min}$-restricted game $(N, \overline{v})$ is $\mathcal{F}$-convex if and only if
the Path, Star, Cycle, Pan, and Adjacent cycles conditions are satisfied.
\end{theorem}

\section{Complexity of inheritance of $\mathcal{F}$-convexity with $\mathcal{P}_{\min}$}
\label{SectionComplexityOfInheritance}

We now prove that to verify Star, Path,
Cycle,
Pan, and Adjacent Cycles conditions we only have to consider
a polynomial number of paths, stars, pans and cycles.
Therefore we can build a polynomial algorithm to decide for a given weighted graph
if there is inheritance of $\mathcal{F}$-convexity from underlying games
to $\mathcal{P}_{\min}$-restricted games.
We assume that graphs are represented by their adjacency matrices.

\begin{proposition}
Let $G= (N,E,w)$ be a weighted graph.
Star condition can be verified
in $O(n^{2})$ time.
\end{proposition}

\begin{proof}
For a given node $i$ in $N$
we denote by $E_{i}$ the set of edges incident to $i$.
If $|E_{i}| \geq 3$ then we select two edges $e_{1}$, $e_{2}$ in $E_{i}$
and replace $E_{i}$ by $E_{i}\setminus \lbrace e_{1}, e_{2} \rbrace$.
If $w_{1} < w_{2}$ we set $\min = w_{1}$, $\max = w_{2}$,
otherwise $\min = w_{2}$, $\max = w_{1}$.
Then we apply the following procedure.
While $E_{i} \not= \emptyset$
we select an edge $e \in E_{i}$.
If $\min = \max$
then if $w(e) \leq \max$
we set $\min = w(e)$,
otherwise we stop.
Otherwise we have $\min < \max$
and if $w(e) \not= \max$ we stop.
Otherwise we replace $E_{i}$ by $ E_{i} \setminus \lbrace e \rbrace$.
If the procedure stops with $E_{i} \not= \emptyset$,
then there is a contradiction to the Star condition.
Otherwise the Star condition is satisfied for the star centered in $i$.
Hence Star condition can be verified for the star centered in $i$
in $O(n)$ time.
As we have to repeat this procedure for all nodes in $N$
we can verify Star condition
in $O(n^{2})$ time.
\end{proof}

We now assume that the Star condition is satisfied by $G$.

Adding an edge to a spanning tree creates a unique cycle,
called a \textit{fundamental cycle}.
Let $T$ be a minimum weight spanning tree of $G$
and $\gamma = \lbrace 1, e_{1}, 2, e_{2}, \ldots, m, e_{m}, m+1 \rbrace$
be an elementary path of $G$.
Let $e_{j_{1}}, e_{j_{2}}, \ldots, e_{j_{k}}$,
with $1 \leq j_{1} < j_{2} < \ldots < j_{k} \leq m$,
be the edges in $E(\gamma) \setminus E(T)$.
For every $e_{j_{l}} \in E(\gamma) \setminus E(T)$,
we denote by $C_{l}$ the fundamental cycle in $G$
associated with $T$ and $e_{j_{l}}$.

\begin{remark}
Of course two fundamental cycles associated with a spanning tree
can be adjacent.
\end{remark}

We suppose that $G$ satisfies the Star condition
and that every fundamental cycle associated with $T$
satisfies the Weak Cycle condition.
(Note that we only have to check for each fundamental cycle
that there are at most two edges of non-maximum weight
and that these edges are adjacent.
And there are exactly $m - n + 1$ fundamental cycles associated with $T$.
)
For each fundamental cycle $C_{l}$,
we denote by $M_{l}$ its maximum edge weight.
As $T$ is a minimum weight spanning tree
we have $w_{j_{l}} = M_{l} \geq w(e)$ for all $e \in E(C_{l})$.
Let $\gamma_{l}$ be the restriction of $\gamma$ to $C_{l}$
and $v_{l}$, $v_{l}^{'}$ its end-vertices.
Let $\gamma_{l}^{'} := C_{l} \setminus \gamma_{l}$ be the other path in $C_{l}$
linking $v_{l}$ and $v_{l}^{'}$
as represented in Figure~\ref{figTreeCycleTwoPaths}.
\begin{figure}[!h]
\centering
\begin{pspicture}(0,-.1)(0,1.8)
\tiny
\begin{psmatrix}[mnode=circle,colsep=.8,rowsep=1]
{$v_{l}$}	& {} & {}	& {$v_{l}^{'}$}\\
{} 	& {}	& {}
\psset{arrows=-, shortput=nab,labelsep={0.05}}
\tiny
\ncline[linecolor=cyan]{1,1}{1,2}
\ncline[linestyle=dashed,linecolor=cyan]{1,2}{1,3}_{$e_{j_{l}}$}
\ncline[linecolor=cyan]{1,3}{1,4}
\psset{labelsep=-.05}
\ncline[linecolor=orange]{2,2}{2,3}
\normalsize
\uput[0](-1,.7){\textcolor{blue}{$C_{l}$}}
\uput[0](-1.8,1.9){\textcolor{cyan}{$\gamma_{l}$}}
\uput[0](-2.8,.7){\textcolor{orange}{$\gamma_{l}^{'}$}}
\ncline[linecolor=orange]{1,1}{2,1}
\ncline[linecolor=orange]{2,1}{2,2}
\ncline[linecolor=orange]{2,3}{1,4}
\end{psmatrix}
\end{pspicture}
\caption{Paths $\gamma_{l}$ and $\gamma_{l}^{'}$ in $C_{l}$.}
\label{figTreeCycleTwoPaths}
\end{figure}
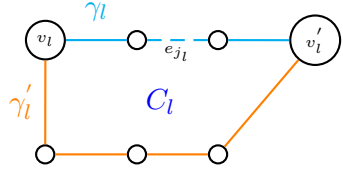

Note that $\gamma_{l}^{'}$ contains at least one edge
otherwise $v_{l} = v_{l}^{'}$ and $\gamma$ is not an elementary path.
We denote by $\gamma_{l,l+1}$ the part of $\gamma$
linking $v_{l}^{'}$ to $v_{l+1}$.
$\gamma_{l,l+1}$ can be reduced to one vertex (if $v_{l}^{'} = v_{l+1}$).
We denote by $\gamma_{0,1}$
(resp. $\gamma_{k, m+1}$)
the part of $\gamma$ linking $1$ to $v_{1}$
(resp. $v_{k}^{'}$ to $m+1$).
$\gamma_{0,1}$ (resp. $\gamma_{k, m+1}$) can also be reduced to one vertex.
Then $\tilde{\gamma} = \gamma_{0,1} \cup \gamma_{1}^{'} \cup \gamma_{1,2}
\cup \gamma_{2}^{'} \cup \cdots \cup \gamma_{k-1,k} \cup \gamma_{k}^{'}
\cup \gamma_{k, m+1}$
corresponds to a path in $T$ linking $1$ to $m+1$.
Let us observe that
if two successive fundamental cycles $C_{l}$, $C_{l+1}$ are adjacent
then $\gamma_{l}{'} \cup \gamma_{l, l+1} \cup \gamma_{l+1}^{'}$
is not a simple path
($\gamma_{l, l+1}$ is reduced to one vertex
but $\gamma_{l}^{'}$ and $\gamma_{l+1}^{'}$ have common edges). 
In this case we delete the common edges to get a simple path
(in fact we get a unique elementary path in $T$)
as represented in Figures~\ref{figTree2CyclesThreePaths}
and \ref{figTree2CyclesThreePaths2}.
Keeping the same notations for simplicity,
we now consider that $\tilde{\gamma}$\label{DefinitionOfTildeGamma} is an elementary path in $T$
linking $1$ to $m+1$.

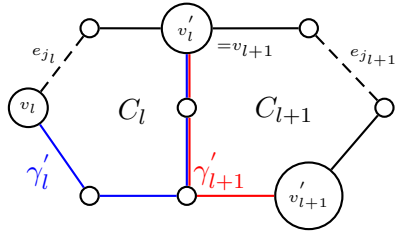
\begin{figure}[!h]
\centering
\begin{pspicture}(0,-.3)(0,2.6)
\tiny
\begin{psmatrix}[mnode=circle,colsep=0.4,rowsep=0.4]
	& {}	&  	&  {$v_{l}^{'}$} & & { }\\
{$v_{l}$} 	& 	&	&  {} & & & { }\\
	& {} 	& 	& {} & & {$v_{l+1}^{'}$}
\psset{arrows=-, shortput=nab,labelsep={0.02}}
\tiny
\ncline[linecolor=blue]{2,1}{3,2}
\ncline[linecolor=blue]{3,2}{3,4}
\ncline[linestyle=dashed]{2,1}{1,2}^{$e_{j_{l}}$}
\ncline[linecolor=blue,shadow=true,shadowsize=1pt,shadowangle=0,shadowcolor=red]{2,4}{3,4}
\ncline[linecolor=blue,shadow=true,shadowsize=1pt,shadowangle=0,shadowcolor=red]{2,4}{1,4}
\ncline{1,2}{1,4}
\ncline{1,4}{1,6}
\ncline[linestyle=dashed]{1,6}{2,7}^{$e_{j_{l+1}}$}
\ncline{2,7}{3,6}
\ncline[linecolor=red]{3,6}{3,4}
\normalsize
\uput[0](-1.5,2){$\scriptscriptstyle{=v_{l+1}}$}
\uput[0](-4,.4){$\textcolor{blue}{\gamma_{l}^{'}}$}
\uput[0](-1.8,.4){$\textcolor{red}{\gamma_{l+1}^{'}}$}
\uput[0](-1,1.2){$C_{l+1}$}
\uput[0](-2.8,1.2){$C_{l}$}
\end{psmatrix}
\end{pspicture}
\caption{We replace $\gamma_{l}{'} \cup \gamma_{l, l+1} \cup \gamma_{l+1}^{'}$
by $(\gamma_{l}{'} \cup \gamma_{l+1}^{'})
\setminus (\gamma_{l}^{'} \cap \gamma_{l+1}^{'})$.}
\label{figTree2CyclesThreePaths}
\end{figure}

\begin{figure}[!h]
\centering
\begin{pspicture}(0,-.3)(0,2.6)
\tiny
\begin{psmatrix}[mnode=circle,colsep=0.4,rowsep=0.4]
	& {}	&  	&  {$v_{l}^{'}$} & & { }\\
{$v_{l}$} 	& 	&	&  {} & & & { }\\
	& {} 	& 	& {} & & {$v_{l+1}^{'}$}
\psset{arrows=-, shortput=nab,labelsep={0.02}}
\tiny
\ncline[linecolor=blue,shadow=true,shadowsize=1pt,shadowangle=-90,shadowcolor=red]{2,1}{2,4}
\ncline[linecolor=red]{2,1}{3,2}
\ncline[linecolor=red]{3,2}{3,4}
\ncline[linestyle=dashed]{2,1}{1,2}^{$e_{j_{l}}$}
\ncline[linecolor=blue,shadow=true,shadowsize=1pt,shadowangle=0,shadowcolor=red]{2,4}{1,4}
\ncline{1,2}{1,4}
\ncline{1,4}{1,6}
\ncline[linestyle=dashed]{1,6}{2,7}^{$e_{j_{l+1}}$}
\ncline{2,7}{3,6}
\ncline[linecolor=red]{3,6}{3,4}
\normalsize
\uput[0](-1.5,2){$\scriptscriptstyle{=v_{l+1}}$}
\uput[0](-3.5,1.7){$\textcolor{blue}{\gamma_{l}^{'}}$}
\uput[0](-1.8,.4){$\textcolor{red}{\gamma_{l+1}^{'}}$}
\uput[0](-1,1.2){$C_{l+1}$}
\uput[0](-2.8,1.8){$C_{l}$}
\end{psmatrix}
\end{pspicture}
\caption{We can have $\gamma_{l}^{'} \subset \gamma_{l+1}^{'}$.}
\label{figTree2CyclesThreePaths2}
\end{figure}
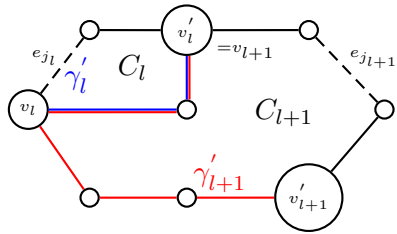

\begin{lemma}
\label{lemIfC_lAndC_l+1AreAdjacentThenM_l=M_l+1}
If $C_{l}$ and $C_{l+1}$ are adjacent
then $M_{l} = M_{l+1}$.
\end{lemma}

\begin{proof}
As the Star condition is satisfied
and as the fundamental cycles associated with $T$
satisfy the Weak Cycle condition,
the result is an immediate consequence
of Lemma~\ref{lemM=maxu(e)=maxu(e)=M'} page~\pageref{lemM=maxu(e)=maxu(e)=M'}.
\end{proof}

\begin{lemma}
\label{lemEitherM_l=w_1Orw_mOrThereExistsAnEdgeeInE(TildeGamma)SuchThatM_l=w(e)}
Let $G= (N,E,w)$ be a weighted graph
satisfying the Star condition.
Let us consider a minimum weight spanning tree $T$
and an elementary path $\gamma = \lbrace 1, e_{1}, 2, e_{2}, \ldots, m, e_{m}, m+1 \rbrace$
in $G$.
Let us assume that every fundamental cycle associated with $T$
satisfies the Weak Cycle condition.
Then for every fundamental cycle $C_{l}$ associated with $T$ and $e_{j_{l}}$
in $E(\gamma) \setminus E(T)$,
with $1 \leq l \leq k$,
either $M_{l} = w_{1}$ or $w_{m}$
or there exists an edge $e \in E(\tilde{\gamma})$ such that
$M_{l} = w(e)$.
\end{lemma}

\begin{proof}
Let us recall that $w_{j_{l}} = M_{l}$.\\
\textbf{
Case 1
}
Let us consider $C_{k}$
and
let us assume $v_{k}^{'} = m+1$,
\emph{i.e.},
$C_{k}$ incident to $m+1$.
Then $e_{m} \in E(\gamma_{k})$.
If $e_{j_{k}} = e_{m}$ or
if $w_{m} = M_{k}$ the result is satisfied.
Hence we can assume
$e_{j_{k}} \not= e_{m}$ 
and $w_{m} < M_{k}$.
Let $\tilde{e}_{1}$ be the (last) edge of
$\tilde{\gamma}$
incident to $m+1$.
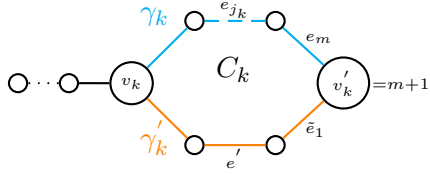
\begin{figure}[!h]
\centering
\begin{pspicture}(0,-.1)(0,2.4)
\tiny
\begin{psmatrix}[mnode=circle,colsep=0.4,rowsep=0.3]
	& 	&  	&  {} & & { }\\
{}  & {} & {$v_{k}$}		&   & & & {$v_{k}^{'}$}\\
	&  	& 	& {} & & {}
\psset{arrows=-, shortput=nab,labelsep={0.02}}
\tiny
\ncline[linestyle=dotted]{2,1}{2,2}
\ncline{2,2}{2,3}
\ncline[linecolor=orange]{2,3}{3,4}
\ncline[linecolor=cyan]{2,3}{1,4}
\ncline[linestyle=dashed,linecolor=cyan]{1,4}{1,6}^{$e_{j_{k}}$}
\ncline[linecolor=cyan]{1,6}{2,7}^{$e_{m}$}
\ncline[linecolor=orange]{2,7}{3,6}^{$\tilde{e}_{1}$}
\ncline[linecolor=orange]{3,4}{3,6}_{$e^{'}$}
\normalsize
\uput[0](1.25,.8){$\scriptscriptstyle{=m+1}$}
\uput[0](-1.8,1.7){$\textcolor{cyan}{\gamma_{k}}$}
\uput[0](-1.8,.1){$\textcolor{orange}{\gamma_{k}^{'}}$}
\uput[0](-.8,1){$C_{k}$}
\end{psmatrix}
\end{pspicture}
\caption{$\tilde{e}_{1}$ in $\gamma^{'}_{k}$.}
\label{figProofTreePath2PathsInC_l}
\end{figure}
\begin{figure}[!h]
\centering
\begin{pspicture}(0,-.1)(0,2.6)
\tiny
\begin{psmatrix}[mnode=circle,colsep=0.4,rowsep=0.3]
					&			&  					& {} & & { }\\
{$v_{l}$} & {}	& {$v_{k}$} & 	 & &		 & {$v_{k}^{'}$}\\
					& {} 		& 					& {} & & {}\\
					& 		&						& {} & &  & {}
\psset{arrows=-, shortput=nab,labelsep={0.02}}
\tiny
\ncline[linestyle=dotted]{2,1}{2,2}
\ncline{2,2}{2,3}
\ncline[linecolor=red,shadow=true,shadowsize=1.5pt,shadowangle=40,shadowcolor=orange]{2,3}{3,4}
\ncline[linecolor=cyan]{2,3}{1,4}
\ncline[linestyle=dashed,linecolor=cyan]{1,4}{1,6}^{$e_{j_{k}}$}
\ncline[linecolor=cyan]{1,6}{2,7}^{$e_{m}$}
\ncline[linecolor=orange,shadow=true,
shadowsize=1.5pt,shadowangle=-40,shadowcolor=red]{3,6}{2,7}_{$\hat{e}$}
\ncline[linecolor=red,shadow=true,shadowsize=1.5pt,shadowangle=90,shadowcolor=orange]{3,4}{3,6}
\ncline[linecolor=red]{2,1}{3,2}
\ncline[linecolor=red]{3,2}{4,4}
\ncline[linecolor=red]{4,4}{4,7}_{$e^{'}$}
\ncline[linecolor=red]{4,7}{2,7}_{$\tilde{e}_{1}$}
\normalsize
\uput[0](.35,1.4){$\scriptscriptstyle{=m+1}$}
\uput[0](-2.8,2.4){$\textcolor{cyan}{\gamma_{k}}$}
\uput[0](-2.4,1.4){$\textcolor{orange}{\gamma_{k}^{'}}$}
\uput[0](-3.2,0){$\textcolor{red}{\gamma_{l}^{'}}$}
\uput[0](-1.7,1.6){$C_{k}$}
\uput[0](-2.7,.6){$C_{l}$}
\end{psmatrix}
\end{pspicture}
\caption{$\tilde{e}_{1}$ in $\gamma^{'}_{l}$ with $l < k$.}
\label{figProofTreePath2PathsInC_l2}
\end{figure}
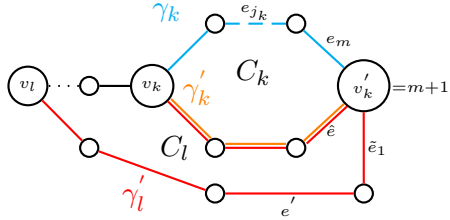
By definition of $\tilde{\gamma}$,
$\tilde{e}_{1}$ belongs
to a path $\gamma_{l}^{'}$ in $C_{l}$
with $1 \leq l \leq k$
as represented in Figure~\ref{figProofTreePath2PathsInC_l}
for $l = k$
and Figure~\ref{figProofTreePath2PathsInC_l2}
for $l <k$.
Note that if $l < k$ then
$C_{l}, C_{l+1}, \ldots, C_{k-1}, C_{k}$ are adjacent
and by Lemma~\ref{lemIfC_lAndC_l+1AreAdjacentThenM_l=M_l+1}
we have $M_{l} = M_{k}$.
If $w(\tilde{e}_{1}) = M_{k}$
we take $e = \tilde{e}_{1}$.
Otherwise
we have $w(\tilde{e}_{1}) < M_{k}$.
Let $\hat{e}$ be the edge in $E(C_{l}) \setminus \lbrace \tilde{e}_{1} \rbrace$
incident to $v_{k}^{'}$.
If $l = k$
then $\hat{e} = e_{m}$
and  we have  $w(\hat{e}) < M_{k}$.
Otherwise
we apply Star condition to $\lbrace e_{m}, \hat{e}, \tilde{e}_{1} \rbrace$
and $w(\hat{e}) < M_{k}$ is still satisfied.
Then the Weak Cycle condition applied to $C_{l}$
implies that there is an edge $e^{'}$ in $E(C_{l})$
adjacent to $\tilde{e}_{1}$
with $w(e^{'}) = M_{k}$.
If $e^{'}$ belongs to $\tilde{\gamma}$ we take $e = e^{'}$.
Otherwise,
if $|E(\tilde{\gamma})| \geq 2$
there is an edge $\tilde{e}_{2}$ in $\tilde{\gamma}$
adjacent to $e^{'}$ and $\tilde{e}_{1}$
as represented in Figure~\ref{figProofTreePath2PathsInC_lSize1}.
Star condition applied to $\lbrace e^{'}, \tilde{e}_{1}, \tilde{e}_{2} \rbrace$
implies $w(\tilde{e}_{2}) = M_{k}$
and we take $e = \tilde{e}_{2}$.
Finally if $|E(\tilde{\gamma})| = 1$
then $v_{l} = v_{1} = 1$.
Hence either $e^{'} = e_{1}$ and the result is satisfied
or $e^{'}$ is adjacent to $e_{1}$
as represented in Figure~\ref{figProofTreePath2PathsInC_lSize2}.
Then Star condition applied to $\lbrace e_{1}, e^{'}, \tilde{e}_{1} \rbrace$
implies $w(e^{'}) = w_{1}$.
Let us observe that
we can have $e^{'} = e_{j_{l}}$
hence in the assumptions,
 Star condition has to be satisfied for $G$
(and not only for $T$).

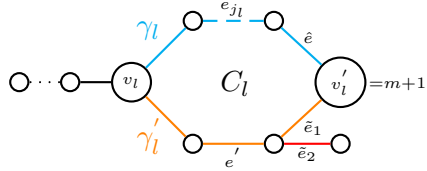
\begin{figure}[!h]
\centering
\begin{pspicture}(0,-.1)(0,2)
\tiny
\begin{psmatrix}[mnode=circle,colsep=0.4,rowsep=0.3]
		&		 &  					&  {}	& & { }\\
{}	& {} & {$v_{l}$}	&  		& &			 & {$v_{l}^{'}$}\\
		&		 &						&	 {} &	& {}   & {}
\psset{arrows=-, shortput=nab,labelsep={0.02}}
\tiny
\ncline[linestyle=dotted]{2,1}{2,2}
\ncline{2,2}{2,3}
\ncline[linecolor=cyan]{2,3}{1,4}
\ncline[linestyle=dashed,linecolor=cyan]{1,4}{1,6}^{$e_{j_{l}}$}
\ncline[linecolor=cyan]{1,6}{2,7}^{$\hat{e}$}
\ncline[linecolor=orange]{2,3}{3,4}
\ncline[linecolor=orange]{3,4}{3,6}_{$e^{'}$}
\ncline[linecolor=red]{3,6}{3,7}_{$\tilde{e}_{2}$}
\ncline[linecolor=orange]{3,6}{2,7}_{$\tilde{e}_{1}$}
\normalsize
\uput[0](.35,.8){$\scriptscriptstyle{=m+1}$}
\uput[0](-2.7,1.5){$\textcolor{cyan}{\gamma_{l}}$}
\uput[0](-2.7,.1){$\textcolor{orange}{\gamma_{l}^{'}}$}
\uput[0](-1.6,.8){$C_{l}$}
\end{psmatrix}
\end{pspicture}
\caption{$|E(\tilde{\gamma})| \geq 2$.}
\label{figProofTreePath2PathsInC_lSize1}
\end{figure}

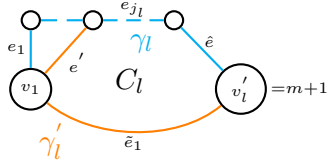
\begin{figure}[!h]
\centering
\begin{pspicture}(0,-.6)(0,1.4)
\tiny
\begin{psmatrix}[mnode=circle,colsep=0.4,rowsep=0.4]
{}	&  {} & & { }\\
{$v_{1}$}		&   & & & {$v_{l}^{'}$}
\psset{arrows=-, shortput=nab,labelsep={0.02}}
\tiny
\ncline[linecolor=cyan]{1,1}{2,1}_{$e_{1}$}
\ncline[linestyle=dashed,linecolor=cyan]{1,1}{1,2}
\ncline[linecolor=orange]{2,1}{1,2}_{$e^{'}$}
\ncline[linestyle=dashed,linecolor=cyan]{1,2}{1,4}^{$e_{j_{l}}$}
\ncline[linecolor=cyan]{1,4}{2,5}^{$\hat{e}$}
\ncarc[linecolor=orange,arcangle=-40]{2,1}{2,5}_{$\tilde{e}_{1}$}
\normalsize
\uput[0](.25,.1){$\scriptscriptstyle{=m+1}$}
\uput[0](-1.6,.7){$\textcolor{cyan}{\gamma_{l}}$}
\uput[0](-2.8,-.6){$\textcolor{orange}{\gamma_{l}^{'}}$}
\uput[0](-1.8,.2){$C_{l}$}
\end{psmatrix}
\end{pspicture}
\caption{$|E(\tilde{\gamma})| = 1$ and $e^{'} \not= e_{1}$.}
\label{figProofTreePath2PathsInC_lSize2}
\end{figure}

\noindent
\textbf{
Case 2
}
Let us now consider $C_{l}$ with $l \leq k$
and $v_{l}^{'} \not= m+1$.
Then there exists an edge $\tilde{e}_{1} \in E(\tilde{\gamma}) \setminus E(C_{l})$
incident to $v_{l}^{'}$.
If $l = k$,
then $\tilde{e}_{1}$ is the edge of $\gamma_{k, m+1}$ incident to $v_{k}^{'}$
as represented in Figure~\ref{figProofTreePath2PathsInC_l+OnePath}.
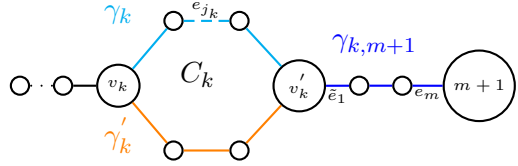
\begin{figure}[!h]
\centering
\begin{pspicture}(0,0)(0,2)
\tiny
\begin{psmatrix}[mnode=circle,colsep=0.3,rowsep=0.2]
	& 	&  	&  {} & & { }\\
{}  & {} & {$v_{k}$}		&   & & & {$v_{k}^{'}$} & {} & {} & {$m+1$}\\
	&  	& 	& {} & & {}
\psset{arrows=-, shortput=nab,labelsep={0.02}}
\tiny
\ncline[linestyle=dotted]{2,1}{2,2}
\ncline{2,2}{2,3}
\ncline[linecolor=orange]{2,3}{3,4}
\ncline[linecolor=cyan]{2,3}{1,4}
\ncline[linestyle=dashed,linecolor=cyan]{1,4}{1,6}^{$e_{j_{k}}$}
\ncline[linecolor=cyan]{1,6}{2,7}
\ncline[linecolor=orange]{2,7}{3,6}
\ncline[linecolor=orange]{3,4}{3,6}
\ncline[linecolor=blue]{2,7}{2,8}_{$\tilde{e}_{1}$}
\ncline[linecolor=blue]{2,8}{2,9}
\ncline[linecolor=blue]{2,9}{2,10}_{$e_{m}$}
\normalsize
\uput[0](-1.8,1.8){$\textcolor{cyan}{\gamma_{k}}$}
\uput[0](-1.8,.2){$\textcolor{orange}{\gamma_{k}^{'}}$}
\uput[0](1.2,1.4){$\textcolor{blue}{\gamma_{k,m+1}}$}
\uput[0](-.8,1){$C_{k}$}
\end{psmatrix}
\end{pspicture}
\caption{$\tilde{e}_{1}$ in $\gamma_{k,m+1}$.}
\label{figProofTreePath2PathsInC_l+OnePath}
\end{figure}

\noindent
If $l < k$,
we can assume that $C_{l}$ and $C_{l+1}$ are not adjacent.
Otherwise Lemma~\ref{lemIfC_lAndC_l+1AreAdjacentThenM_l=M_l+1}
implies $M_{l} = M_{l+1}$
and we can replace $C_{l}$ by $C_{l+1}$.
If $v_{l}^{'} \not= v_{l+1}$,
$\tilde{e}_{1}$ is the edge of $\gamma_{l,l+1}$
incident to $v_{l}^{'}$
as represented in Figure~\ref{figProofTreePath2PathsInC_l+OnePath+OneCycle}.

\begin{figure}[!h]
\centering
\begin{pspicture}(0,0)(0,1.9)
\tiny
\begin{psmatrix}[mnode=circle,colsep=0.3,rowsep=0.2]
	 & 		&  					& {} & & {}		&								&		 &		&		 & {} & 	 &	{}\\
{} & {} & {$v_{l}$}	&    & & 			& {$v_{l}^{'}$} & {} & {} & {} &		&		 & 		 & {} & {}\\
	 &  	& 					& {} & & {}		&								&		 &		&		 & {}	&		 &	{}				
\psset{arrows=-, shortput=nab,labelsep={0.02}}
\tiny
\ncline[linestyle=dotted]{2,1}{2,2}
\ncline{2,2}{2,3}
\ncline[linecolor=orange]{2,3}{3,4}
\ncline[linecolor=cyan]{2,3}{1,4}
\ncline[linestyle=dashed,linecolor=cyan]{1,4}{1,6}^{$e_{j_{l}}$}
\ncline[linecolor=cyan]{1,6}{2,7}
\ncline[linecolor=orange]{2,7}{3,6}
\ncline[linecolor=orange]{3,4}{3,6}
\ncline[linecolor=blue]{2,7}{2,8}_{$\tilde{e}_{1}$}
\ncline[linecolor=blue]{2,8}{2,9}
\ncline[linecolor=blue]{2,9}{2,10}
\ncline[linecolor=cyan]{2,10}{1,11}
\ncline[linestyle=dashed,linecolor=cyan]{1,11}{1,13}
\ncline[linecolor=cyan]{1,13}{2,14}
\ncline[linecolor=orange]{2,10}{3,11}
\ncline[linecolor=orange]{3,11}{3,13}
\ncline[linecolor=orange]{3,13}{2,14}
\ncline[linestyle=dotted]{2,14}{2,15}
\normalsize
\uput[0](-5.6,1.6){$\textcolor{cyan}{\gamma_{l}}$}
\uput[0](-5.6,0){$\textcolor{orange}{\gamma_{l}^{'}}$}
\uput[0](-2.6,1.2){$\textcolor{blue}{\gamma_{l,l+1}}$}
\uput[0](-.9,.8){$C_{l+1}$}
\uput[0](-4.8,.8){$C_{l}$}
\end{psmatrix}
\end{pspicture}
\caption{$v_{l}^{'} \not= v_{l+1}$, $\tilde{e}_{1}$ in $\gamma_{l,l+1}$.}
\label{figProofTreePath2PathsInC_l+OnePath+OneCycle}
\end{figure}

\noindent
If $v_{l}^{'} = v_{l+1}$,
$\tilde{e}_{1}$ is the first edge of $\gamma_{l+1}^{'}$ incident to $v_{l}^{'}$
as represented in Figure~\ref{figProofTreePath2IncidentCyclesC_lC_l+1}.
\begin{figure}[!h]
\centering
\begin{pspicture}(0,0)(0,2)
\tiny
\begin{psmatrix}[mnode=circle,colsep=0.3,rowsep=0.2]
	 & 		&  					& {} & & {}		&								& {} & 	 &	{}\\
{} & {} & {$v_{l}$}	&    & & 			& {$v_{l}^{'}$} &		&		 & 		 & {} & {}\\
	 &  	& 					& {} & & {}		&								& {}	&		 &	{}				
\psset{arrows=-, shortput=nab,labelsep={0.02}}
\tiny
\ncline[linestyle=dotted]{2,1}{2,2}
\ncline{2,2}{2,3}
\ncline[linecolor=orange]{2,3}{3,4}
\ncline[linecolor=cyan]{2,3}{1,4}
\ncline[linestyle=dashed,linecolor=cyan]{1,4}{1,6}^{$e_{j_{l}}$}
\ncline[linecolor=cyan]{1,6}{2,7}
\ncline[linecolor=orange]{2,7}{3,6}
\ncline[linecolor=orange]{3,4}{3,6}
\ncline[linecolor=cyan]{2,7}{1,8}
\ncline[linestyle=dashed,linecolor=cyan]{1,8}{1,10}^{$e_{j_{l+1}}$}
\ncline[linecolor=cyan]{1,10}{2,11}
\ncline[linecolor=orange]{2,7}{3,8}^{$\tilde{e}_{1}$}
\ncline[linecolor=orange]{3,8}{3,10}
\ncline[linecolor=orange]{3,10}{2,11}
\ncline[linestyle=dotted]{2,11}{2,12}
\normalsize
\uput[0](-4.1,1.6){$\textcolor{cyan}{\gamma_{l}}$}
\uput[0](-4.1,.1){$\textcolor{orange}{\gamma_{l}^{'}}$}
\uput[0](-1.8,1.6){$\textcolor{cyan}{\gamma_{l+1}}$}
\uput[0](-1.8,.1){$\textcolor{orange}{\gamma_{l+1}^{'}}$}
\uput[0](-.9,.8){$C_{l+1}$}
\uput[0](-3,.8){$C_{l}$}
\end{psmatrix}
\end{pspicture}
\caption{$v_{l}^{'} = v_{l+1}$, $\tilde{e}_{1}$ in $\gamma_{l+1}^{'}$.}
\label{figProofTreePath2IncidentCyclesC_lC_l+1}
\end{figure}
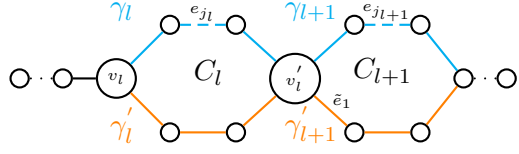

\noindent
Let us denote by $\tilde{\gamma}_{1,l}$ the part of $\tilde{\gamma}$
linking $1$ to $v_{l}^{'}$.
Let $\tilde{e}_{2}$
be the edge of $\tilde{\gamma}_{1,l}$
incident to $v_{l}^{'}$
as represented in Figure~\ref{figProofTreePath1Cycle}
with $\tilde{e}_{2}$ in $\gamma_{l}^{'}$.
If $\tilde{e}_{2} \notin \gamma_{l}^{'}$
then there exists a fundamental cycle $C_{r}$ for some $r$, with $1 \leq r \leq l$,
such that $\tilde{e}_{2} \in \gamma_{r}^{'}$
and $C_{r}, C_{r+1}, \ldots, C_{l-1}, C_{l}$ are adjacent.
By Lemma~\ref{lemIfC_lAndC_l+1AreAdjacentThenM_l=M_l+1}
we have $M_{r} = M_{l}$.
Therefore we can suppose $\tilde{e}_{2} \in \gamma_{l}^{'}$
replacing $C_{l}$ by $C_{r}$ if necessary.
Note that in this last case $C_{l}$ may be adjacent to $C_{l+1}$
but $\tilde{e}_{1} \notin E(C_{l})$.
Let $\hat{e}$
be the edge of $\gamma_{l}$
incident to $v_{l}^{'}$
as represented in Figure~\ref{figProofTreePath1Cycle}.

\begin{figure}[!h]
\centering
\begin{pspicture}(0,0)(0,2)
\tiny
\begin{psmatrix}[mnode=circle,colsep=0.3,rowsep=0.2]
	 & 		&  					& {} & & {}\\
{} & {} & {$v_{l}$}	&    & & 			& {$v_{l}^{'}$} & {}\\
	 &  	& 					& {} & & {}			
\psset{arrows=-, shortput=nab,labelsep={0.02}}
\tiny
\ncline[linestyle=dotted]{2,1}{2,2}
\ncline{2,2}{2,3}
\ncline[linecolor=orange]{2,3}{3,4}
\ncline[linecolor=cyan]{2,3}{1,4}
\ncline[linestyle=dashed,linecolor=cyan]{1,4}{1,6}^{$e_{j_{l}}$}
\ncline[linecolor=cyan]{1,6}{2,7}^{$\hat{e}$}
\ncline[linecolor=orange]{2,7}{3,6}^{$\tilde{e}_{2}$}
\ncline[linecolor=orange]{3,4}{3,6}
\ncline{2,7}{2,8}^{$\tilde{e}_{1}$}
\normalsize
\uput[0](-1.6,1.6){$\textcolor{cyan}{\gamma_{l}}$}
\uput[0](-1.6,.1){$\textcolor{orange}{\gamma_{l}^{'}}$}
\uput[0](-.7,.8){$C_{l}$}
\end{psmatrix}
\end{pspicture}
\caption{$\hat{e}$ in $\gamma_{l}$, $\tilde{e}_{2}$ in $\gamma_{l}^{'}$.}
\label{figProofTreePath1Cycle}
\end{figure}
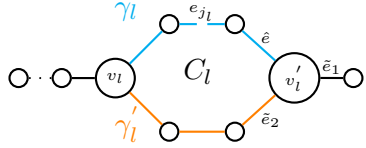

\noindent
If $w(\tilde{e}_{2}) = M_{l}$,
we take $e = \tilde{e}_{2}$.
Otherwise we have $w(\tilde{e}_{2}) <  M_{l}$.
If $w(\hat{e}) = M_{l}$,
Star condition applied to $\lbrace \hat{e}, \tilde{e}_{1}, \tilde{e}_{2} \rbrace$
implies $w(\tilde{e}_{1}) = M_{l}$
and we can take $e = \tilde{e}_{1}$.
If now $w(\hat{e})$
and $w(\tilde{e}_{2}) < M_{l}$,
then the Weak Cycle condition applied to $C_{l}$
implies that there is an edge $e^{'}$ in $E(C_{l})$
adjacent to $\tilde{e}_{2}$
with $w(e^{'}) = M_{l}$.
If $e^{'}$ belongs to $\tilde{\gamma}$ we take $e = e^{'}$.
Otherwise,
if $|E(\tilde{\gamma}_{1,l})| \geq 2$
there is an edge $\tilde{e}_{3}$ in $\tilde{\gamma}$
adjacent to $e^{'}$ and $\tilde{e}_{2}$
as represented in Figure~\ref{figProofTreePath2PathsInC_l3}.
Star condition applied to $\lbrace e^{'}, \tilde{e}_{2}, \tilde{e}_{3} \rbrace$
implies $w(\tilde{e}_{3}) = M_{l}$
and we take $e = \tilde{e}_{3}$.
Finally if $|E(\tilde{\gamma}_{1,l})| = 1$
then $v_{l}= v_{1} = 1$.
Then either $e^{'} = e_{1} $ and the result is satisfied
or $e^{'}$ is adjacent to $e_{1}$ as represented in Figure~\ref{figProofTreePath2PathsInC_l4}.
Then
as $w(\tilde{e}_{2}) < w(e^{'})$,
Star condition applied to $\lbrace e^{'}, e_{1}, \tilde{e}_{2} \rbrace$
implies $w(e^{'}) = w_{1}$.
Note that the situation is similar to Case 1 with $l < k$
replacing $\tilde{e}_{1}$ by $\tilde{e}_{2}$.

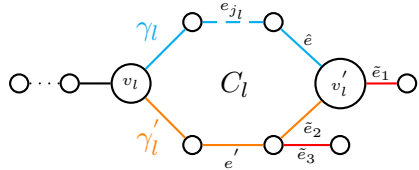
\begin{figure}[!h]
\centering
\begin{pspicture}(0,-.1)(0,2)
\tiny
\begin{psmatrix}[mnode=circle,colsep=0.4,rowsep=0.3]
		&		 &  					&  {}	& & { }\\
{}	& {} & {$v_{l}$}	&  		& &			 & {$v_{l}^{'}$}& {}\\
		&		 &						&	 {} &	& {}   & {}
\psset{arrows=-, shortput=nab,labelsep={0.02}}
\tiny
\ncline[linestyle=dotted]{2,1}{2,2}
\ncline{2,2}{2,3}
\ncline[linecolor=cyan]{2,3}{1,4}
\ncline[linestyle=dashed,linecolor=cyan]{1,4}{1,6}^{$e_{j_{l}}$}
\ncline[linecolor=cyan]{1,6}{2,7}^{$\hat{e}$}
\ncline[linecolor=orange]{2,3}{3,4}
\ncline[linecolor=orange]{3,4}{3,6}_{$e^{'}$}
\ncline[linecolor=red]{3,6}{3,7}_{$\tilde{e}_{3}$}
\ncline[linecolor=orange]{3,6}{2,7}_{$\tilde{e}_{2}$}
\ncline[linecolor=red]{2,7}{2,8}^{$\tilde{e}_{1}$}
\normalsize
\uput[0](-2.7,1.5){$\textcolor{cyan}{\gamma_{l}}$}
\uput[0](-2.7,.1){$\textcolor{orange}{\gamma_{l}^{'}}$}
\uput[0](-1.6,.8){$C_{l}$}
\end{psmatrix}
\end{pspicture}
\caption{$|E(\tilde{\gamma}_{1,l})| \geq 2$, $\tilde{e}_{3}$ in $\tilde{\gamma}$
adjacent to $e^{'}$ and $\tilde{e}_{2}$.}
\label{figProofTreePath2PathsInC_l3}
\end{figure}

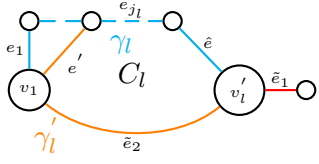
\begin{figure}[!h]
\centering
\begin{pspicture}(0,-.6)(0,1.4)
\tiny
\begin{psmatrix}[mnode=circle,colsep=0.4,rowsep=0.4]
{}	&  {} & & { }\\
{$v_{1}$}		&   & & & {$v_{l}^{'}$} & {}
\psset{arrows=-, shortput=nab,labelsep={0.02}}
\tiny
\ncline[linecolor=cyan]{1,1}{2,1}_{$e_{1}$}
\ncline[linestyle=dashed,linecolor=cyan]{1,1}{1,2}
\ncline[linecolor=orange]{2,1}{1,2}_{$e^{'}$}
\ncline[linestyle=dashed,linecolor=cyan]{1,2}{1,4}^{$e_{j_{l}}$}
\ncline[linecolor=cyan]{1,4}{2,5}^{$\hat{e}$}
\ncarc[linecolor=orange,arcangle=-40]{2,1}{2,5}_{$\tilde{e}_{2}$}
\ncline[linecolor=red]{2,5}{2,6}^{$\tilde{e}_{1}$}
\normalsize
\uput[0](-2.6,.6){$\textcolor{cyan}{\gamma_{l}}$}
\uput[0](-3.6,-.6){$\textcolor{orange}{\gamma_{l}^{'}}$}
\uput[0](-2.5,.2){$C_{l}$}
\end{psmatrix}
\end{pspicture}
\caption{$|E(\tilde{\gamma}_{1,l})| = 1$, $v_{l} = v_{1} = 1$ and $e^{'} \not= e_{1}$.}
\label{figProofTreePath2PathsInC_l4}
\end{figure}

\end{proof}

Using for instance Prim's algorithm,
we can find a spanning tree $T$ of minimum weight in $O(n^{2})$ time.
The following proposition shows that the Path condition
is satisfied for $G$
if it is satisfied for $T$
and if all fundamental cycles associated with $T$
satisfy the Weak Cycle condition,
and therefore can be verified in polynomial time.

\begin{proposition}
Let $G=(N,E,w)$ be a weighted connected graph
satisfying the Star condition.
Let $T$ be a minimum weight spanning tree of $G$.
If $T$ satisfies the Path condition
and if each fundamental cycle associated with $T$
satisfies the Weak Cycle condition,
then $G$ also satisfies the Path condition. 
Therefore the Path condition can be verified in $O(n^{3})$ time.
\end{proposition}

\begin{proof}
Let $\gamma = \lbrace 1, e_{1}, 2, e_{2}, \ldots, m, e_{m}, m+1 \rbrace$ be a path of $G$.
We have to prove that for every $j$, $2 \leq j \leq m-1$,
we have $w_{j} \leq \max (w_{1}, w_{m})$.

Let us consider $e \in E(\gamma) \setminus E(T)$
and $C$ the fundamental cycle associated with $T$ and $e$.
Let $M$ be the maximum edge weight in $C$.
We necessarily have $w(e) = M$ otherwise $T$ is not a minimum weight
spanning tree.
We have to prove:
\begin{equation}
\label{eqProofM<=Max(w_1,w_m}
M\leq \max (w_{1}, w_{m}).
\end{equation}
Let $\tilde{\gamma}$ be the path in $T$ related to $\gamma$
linking $1$ to $m+1$ 
defined page~\pageref{DefinitionOfTildeGamma}. 
Lemma~\ref{lemEitherM_l=w_1Orw_mOrThereExistsAnEdgeeInE(TildeGamma)SuchThatM_l=w(e)}
implies that either $M = w_{1}$ or $w_{m}$
and then (\ref{eqProofM<=Max(w_1,w_m}) is obviously satisfied
or there exists $\tilde{e} \in E(\tilde{\gamma})$
such that $w(\tilde{e}) = M$.
Let $e_{1}^{'}$ (resp. $e_{m}^{'}$)
be the edge of $\tilde{\gamma}$
incident to $1$ (resp. $m+1$).
Path condition applied to $\tilde{\gamma}$ implies
$w(\tilde{e}) \leq \max(w_{1}^{'}, w_{m}^{'})$.
If $e_{1}$ (resp. $e_{m}$)
is not an edge of $T$,
then it forms a fundamental cycle containing $e_{1}^{'}$ (resp. $e_{m}^{'}$)
as represented in Figure~\ref{figLemmaTreeAndPath4}
and as $T$ is a minimum weight spanning tree
we have $w_{1}^{'} \leq w_{1}$ (resp. $w_{m}^{'} \leq w_{m}$).
If $e_{1}$ (resp. $e_{m}$) is an edge of $T$,
then Path condition applied to $e_{1} \cup \tilde{\gamma}$ (resp. $e_{m} \cup \tilde{\gamma}$)
implies $w(\tilde{e}) \leq \max(w_{1}, w_{m}^{'})$
(resp. $w(\tilde{e}) \leq \max(w_{1}^{'}, w_{m})$).
If both $e_{1}$ and $e_{m}$ are edges of $T$,
then Path condition applied to $e_{1} \cup \tilde{\gamma} \cup e_{m}$
implies $w(\tilde{e}) \leq \max(w_{1}, w_{m})$.
Hence in every case, $w(\tilde{e}) \leq \max (w_{1}, w_{m})$
and therefore (\ref{eqProofM<=Max(w_1,w_m}) is satisfied.
As $M$ is the maximum edge-weight in $C$,
we get $w(e) \leq \max(w_{1}, w_{m})$ for all $e \in E(C)$.
Therefore we also have $w(e) \leq \max(w_{1}, w_{m})$ for all $e$
in a fundamental cycle associated with $T$.
Let $e$ be a remaining edge in $E(\gamma) \cap E(T)$.
$e$ is necessarily in $\tilde{\gamma}$
(more precisely in one of the subpaths $\gamma_{0,1}, \gamma_{1,2}, \ldots, \gamma_{k,m+1}$)
and therefore satisfies $w(e) \leq \max(w_{1}, w_{m})$.

\begin{figure}[!h]
\centering
\begin{pspicture}(0,-.3)(0,2)
\tiny
\begin{psmatrix}[mnode=circle,colsep=1,rowsep=1]
{$1$}	& {} 	& {}	& {} & {m$$} \\
{} 	& {}	& {}	& {}
\psset{arrows=-, shortput=nab,labelsep={.1}}
\tiny
\ncline[linecolor=blue]{1,1}{1,2}^{$e_{1}$}
\ncline[linecolor=black,shadow=true,shadowsize=1pt,shadowangle=90,shadowcolor=blue]
{1,4}{1,5}^{$e_{m}$}_{$e_{m}^{'}$}
\psset{labelsep=.05}
\ncline{2,2}{1,2}
\ncline[linecolor=black,shadow=true,shadowsize=1pt,shadowangle=90,shadowcolor=blue]{1,2}{1,3}
\ncline{1,3}{2,3}
\ncline[linecolor=blue]{1,3}{1,4}
\ncline{2,3}{2,4}
\ncline{2,4}{1,4}
\ncline{1,1}{2,1}_{$e_{1}^{'}$}
\ncline{2,1}{2,2}
\end{psmatrix}
\normalsize
\uput[0](-6.4,1.2){$\textcolor{blue}{\gamma}$}
\uput[0](-6.4,.2){$\tilde{\gamma}$}
\end{pspicture}
\caption{$\gamma$ and $\tilde{\gamma}$, $e_{1} \notin E(T)$, $e_{m} \in E(T)$.}
\label{figLemmaTreeAndPath4}
\end{figure}
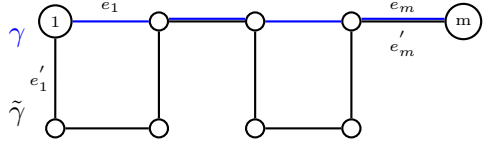

\noindent
Let us now analyze the complexity.
Prim's algorithm provides
a minimum weight spanning tree
in time $O(n^{2})$.


Then we verify Path condition for $T$.
For each vertex $v \in V$,
we build the arborescence in $T$ rooted in $v$.
It can be done using a breadth first search in $O(n^{2})$.
During the search we can verify the Path condition.
Indeed we only have to compare the weight of an edge added to the arborescence
to the weight of its predecessor in the arborescence.
If there is an increase of edge-weights followed by a decrease
then it contradicts the Path condition
(We can assign a label $\lbrace +, -\rbrace$ to each vertex entering the arborescence,
and therefore associated with the entering edge,
corresponding to an increase or decrease,
and we compare it to the label of the preceding node).
Hence we can verify the Path condition on $T$ in $ O(n^{3})$ time.

For a given fundamental cycle associated with $T$,
the Cycle condition can be verified in $O(n)$ time.
Indeed,
as $T$ contains at most $n-1$ edges,
a fundamental cycle contains at most $n$ edges.
Then an edge of minimum weight can be found after
at most $n-1$ comparisons.
An edge of minimum weight
among the $n-1$ remaining edges
can be found using at most $n-2$ comparisons.
Actually we only have to consider the two edges adjacent to the minimum
weight edge previously found.
Finally we only have to verify the equality of edge-weights
for the remaining edges.
It requires at most $n-3$ comparisons.
There exists $m - (n-1)$ edges in $E \setminus E(T)$
and therefore $m - (n-1)$ fundamental cycles associated with $T$.
Hence we can verify the Cycle condition for all fundamental cycles
associated with $T$ in $O(nm)$. 

Finally Path condition for $G$ can be verified in $O(n^{2} +  n^{3} +nm)$ time
and therefore in $O(n^{3})$ time (as $m \leq \frac{n(n-1)}{2}$).
\end{proof}

We now prove that the Adjacent cycles condition can be verified
in $O(n^{6})$ time.
We first establish that we only need to verify the condition on specific cycles
built with shortest paths.

\begin{lemma}
\label{lemMINADJCYCLESCOND}
Let $G = (N,E,w)$ be a weighted graph
satisfying the Star and Path
(and therefore Weak and Intermediary Cycle) conditions.
Let $e_{1} = \lbrace 1, 2 \rbrace$ and $e_{2} = \lbrace 2, 3 \rbrace$
be two adjacent edges of $G$.
Let us define
$\tilde{E} := \lbrace e \in E; \; w(e) > \max (w_{1}, w_{2}) \rbrace$
and $\tilde{G} := (N \setminus \lbrace 2 \rbrace, \tilde{E}, w_{|\tilde{E}})$.
Let $\gamma$ be a shortest path in $\tilde{G}$, if it exists,
linking $1$ and $3$.
Let us consider the following claims:
\begin{enumerate}[1)]
\item
\label{LemItem2cyclesCC'}
There exist two simple cycles $C$ and $C^{'}$ in $G$
satisfying conditions
\ref{enumPropV(C)-V(C)^{'}nonempty},
\ref{enumPropCatmost1non-maxweightchord},
\ref{enumPropCC^{'}nomaxweightchord},
\ref{enumPropCandC'HaveNoCommonChord}
of the Adjacent Cycles condition,
and such that:
\begin{enumerate}[(e)]
\item
\label{LemItemItemCC^{'}twocommonnonmaxweightedges}
$e_{1}$ and $e_{2}$ are two common non-maximum weight edges
of $C$ and $C^{'}$.
\end{enumerate}
\item
\label{LemItem2cyclesCC'2}
There exist two simple cycles $\tilde{C}$ and $\tilde{C}^{'}$
satisfying conditions
\ref{enumPropV(C)-V(C)^{'}nonempty},
\ref{enumPropCatmost1non-maxweightchord},
\ref{enumPropCC^{'}nomaxweightchord},
\ref{enumPropCandC'HaveNoCommonChord}
of the Adjacent Cycles condition,
and such that:
\begin{enumerate}
\setcounter{enumii}{5}
\item
\label{LemItemItemCC{'}twocommonmaxweightedges2}
$\tilde{C}$ and $\tilde{C}^{'}$ have two common non-maximum weight edges
$\tilde{e}_{1} = \lbrace \tilde{1}, 2 \rbrace$, $\tilde{e}_{2} = \lbrace 2, \tilde{3} \rbrace$.
\item \label{LemItemItemC2shortestpathG213iV13C'2G2i13}
$\tilde{C} \setminus \lbrace 2 \rbrace$ corresponds to a shortest path $\tilde{\gamma}$
in $\tilde{G}$
linking $\tilde{1}$ and $\tilde{3}$.
$\tilde{\gamma}$ is a subpath of $\gamma$
and for some $i \in V(\tilde{\gamma}) \setminus \lbrace \tilde{1}, \tilde{3} \rbrace$,
$\tilde{C}^{'} \setminus \lbrace 2 \rbrace$ corresponds to a shortest path $\gamma^{'}$
in $\tilde{G}_{N\setminus \lbrace 2, i \rbrace}$ linking $\tilde{1}$ and $\tilde{3}$.
\end{enumerate}
\end{enumerate}
Then Claim~\ref{LemItem2cyclesCC'} implies Claim~\ref{LemItem2cyclesCC'2}.
Moreover for given adjacent edges $e_{1}$ and $e_{2}$ of $G$,
Claim~\ref{LemItem2cyclesCC'2}
can be verified in  $O(n^{3})$ time.
\end{lemma}

\begin{remark}
It can happen that $\tilde{\gamma} = \gamma$
and then $1 = \tilde{1}$
and $3 = \tilde{3}$.
\end{remark}

\begin{remark}
Let us observe that $|V(C)| \geq 4$ and $|V(C^{'})| \geq 4$,
otherwise
by the Weak Cycle condition
$C$ or $C^{'}$ would have a maximum weight chord.
Let us also observe that by the Weak Cycle condition
and Lemma~\ref{lemM=maxu(e)=maxu(e)=M'} page~\pageref{lemM=maxu(e)=maxu(e)=M'},
all edges of $\gamma$ and $\gamma^{'}$ have the same weight $M$.
\end{remark}

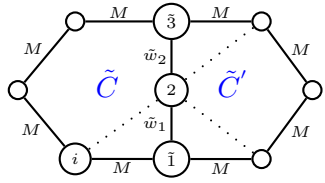
\begin{figure}[!h]
\centering
\begin{pspicture}(0,-.2)(0,2.2)
\tiny
\begin{psmatrix}[mnode=circle,colsep=0.4,rowsep=0.4]
	& {}	&  	&  {$\tilde{3}$} & & { }\\
{} 	& 	&	&  {$2$} & & & { }\\
	& {$i$} 	& 	& {$\tilde{1}$} & & { }
\psset{arrows=-, shortput=nab,labelsep={0.02}}
\tiny
\ncline[linestyle=dotted]{3,2}{2,4}
\ncline[linestyle=dotted]{2,4}{1,6}
\ncline[linestyle=dotted]{2,4}{3,6}
\ncline{2,1}{3,2}_{$M$}
\ncline{3,2}{3,4}_{$M$}
\ncline{2,1}{1,2}^{$M$}
\ncline{2,4}{3,4}_{$\tilde{w}_{1}$}
\ncline{2,4}{1,4}^{$\tilde{w}_{2}$}
\ncline{1,2}{1,4}^{$M$}
\ncline{1,4}{1,6}^{$M$}
\ncline{1,6}{2,7}^{$M$}
\ncline{2,7}{3,6}^{$M$}
\ncline{3,6}{3,4}^{$M$}
\normalsize
\uput[0](-.6,1){\textcolor{blue}{$\tilde{C}^{'}$}}
\uput[0](-2.2,1){\textcolor{blue}{$\tilde{C}$}}
\end{psmatrix}
\end{pspicture}
\caption{$i \in \tilde{\gamma} = \tilde{C} \setminus \lbrace 2 \rbrace$
and $i \not \in \gamma^{'} = \tilde{C}^{'} \setminus \lbrace 2 \rbrace$.}
\label{fig2adjacentcycles13}
\end{figure}

\begin{proof}[Proof of Lemma~\ref{lemMINADJCYCLESCOND}]
We assume w.l.o.g. $w_{1} \leq w_{2}$.
Let us assume that Claim~\ref{LemItem2cyclesCC'} is satisfied.
Lemma~\ref{lemM=maxu(e)=maxu(e)=M'} implies
$\max_{e \in \hat{E}(C)} w(e) = \max_{e \in \hat{E}(C^{'})} w(e) = M$.
Let $\gamma$ be a shortest path in $\tilde{G}$
linking $1$ to $3$.
Then $\hat{C} = \lbrace 2, e_{1}, 1 \rbrace \cup \gamma \cup \lbrace 3, e_{2}, 2  \rbrace$
is a simple cycle in $G$ adjacent to $C$
or $\hat{C} = C$.
The Weak Cycle condition
and Lemma~\ref{lemM=maxu(e)=maxu(e)=M'} imply that $e_{1}$ and $e_{2}$
are edges of non-maximum weight in $\hat{C}$ and all edges of $\gamma$ have weight
$M = \max_{e \in \hat{E}(C)} w(e)$.
As $M > w_{2} \geq w_{1}$,
the Intermediary Cycle condition
implies w(e) = M for all $e \in \hat{E}(\hat{C})$ non incident to $2$
and $w(e) = w_{2}$ or $w(e) < w_{1} = w_{2}$ for any chord incident to $2$.
Hence $\hat{C}$ has no chord of maximum weight $M$
otherwise $\gamma$ is not a shortest path.

\noindent
\textbf{
Case 1
}
Let us assume that $\hat{C}$ has no chord.
Then we take $\tilde{C} = \hat{C}$.
If $V(\tilde{C}) \setminus V(C^{'}) \not= \emptyset$
we select a node $i$ in $V(\tilde{C}) \setminus V(C^{'})$.
Otherwise we have $V(\tilde{C}) \subseteq V(C^{'})$.
$V(\tilde{C}) \subset V(C^{'})$ would imply that
$C^{'}$ has a chord of maximum weight $M$, a contradiction.
Hence $V(\tilde{C}) = V(C^{'})$
and in this case we have
$V(\tilde{C}) \setminus V(C) = V(C^{'}) \setminus V(C) \not= \emptyset$.
Then we select a node $i$ in $V(\tilde{C}) \setminus V(C)$
(we can interchange $C$ and $C^{'}$).
By construction of $i$
there exists a path in $\tilde{G}_{N \setminus \lbrace 2, i \rbrace}$
linking $1$ to $3$
(\emph{i.e.}, a path in $C^{'}$ or $C$).
Therefore there exists a shortest path $\gamma^{'}$
in $\tilde{G}_{N \setminus \lbrace 2, i \rbrace}$ linking $1$ to $3$.
Then the cycle
$\tilde{C}^{'} = \lbrace 2, e_{1}, 1\rbrace \cup \gamma^{'} \cup \lbrace 3, e_{2}, 2 \rbrace$ 
is convenient.
We have $V(\tilde{C}) \setminus V(\tilde{C}^{'}) \not= \emptyset$
as $i \in V(\tilde{C}) \setminus V(\tilde{C}^{'})$.
If $V(\tilde{C}^{'}) \setminus V(\tilde{C}) = \emptyset$
then $V(\tilde{C}^{'}) \subseteq V(\tilde{C})$.
$V(\tilde{C}^{'}) = V(\tilde{C})$ is not possible as $i \in V(\tilde{C}) \setminus V(\tilde{C}^{'})$.
Then $V(\tilde{C}^{'}) \subset V(\tilde{C})$
and $\gamma^{'}$ is a path strictly shorter than $\gamma$
linking $1$ to $3$ in $\tilde{G}$,
a contradiction.
$\tilde{C}$ and $\tilde{C}^{'}$ have no maximum weight chord
because $\gamma$ and $\gamma^{'}$ are shortest path linking $1$ and $3$
in respectively $\tilde{G}$
and $\tilde{G}_{N \setminus \lbrace 2, i \rbrace}$.
$\tilde{C}$ has no chord.
$\tilde{C}$ and $\tilde{C}^{'}$ satisfy conditions
\ref{enumPropV(C)-V(C)^{'}nonempty} to \ref{enumPropCandC'HaveNoCommonChord},
\ref{LemItemItemCC{'}twocommonmaxweightedges2}
and \ref{LemItemItemC2shortestpathG213iV13C'2G2i13}.

If $e_{1}$ and $e_{2}$ are given,
we can obtain a shortest path $\gamma$,
if it exists,
using a Breadth First Search algorithm in $O(n^{2})$ time.
Then to decide if $\hat{C}$ has a chord or not,
take $O(n)$ time.
Indeed $\hat{C}$ cannot have a maximum weight chord
otherwise $\gamma$ is not a shortest path,
and any non maximum weight chord is incident to $2$.
Therefore we only have to check in $G$
if there exists a neighbor of $2$
in $\gamma \setminus \lbrace 1, 3 \rbrace$.
In the worst case, we have to consider $n-3$ vertices.
For a given node $i \in \gamma \setminus \lbrace 1, 3 \rbrace$
checking the existence of $\lbrace 2, i \rbrace$
in the adjacency matrix 
is in $O(1)$ time.
Therefore we can check the existence of a chord in $O(n)$ time.
If $\hat{C}$ is chordless
then for a given $i \in V(\gamma)$ 
we look for a shortest path $\gamma^{'}$ in $\tilde{G}_{N \setminus \lbrace 2, i \rbrace}$
linking $1$ and $3$ using the BFS algorithm in $O(n^{2})$ time.
In the worst case we repeat this for all $i \in V(\gamma)$.
Hence in this case ($\hat{C}$ without chord)
we need  $O(n^{3})$ time
to decide if Claim~\ref{LemItem2cyclesCC'2} is satisfied.

\noindent
\textbf{
Case 2
}
Let us now assume that $\hat{C}$ has a chord $\hat{e}$.
Then it is necessarily a non-maximum weight chord
incident to $2$
with $w(\hat{e}) \leq w_{2}$ (by the Intermediary Cycle condition).
We choose arbitrarily the end-vertex $i$
of any chord $\hat{e}_{0} = \lbrace 2, i \rbrace$ of $\hat{C}$.
We select two other edges in $\hat{E}(\hat{C})$,
$\hat{e}_{1} = \lbrace 2, k \rbrace$
and $\hat{e}_{2} = \lbrace 2, j \rbrace$,
with $3 \leq j < i < k$,
which are as close as possible from $\hat{e}_{0}$,
\emph{i.e.}, $j$ is the maximum possible index
for such an edge $\hat{e}_{2}$
and $k$ is the minimum possible index for such an edge $\hat{e}_{1}$.
Hence if $j=3$ (resp. $k=1$) then $\hat{e}_{2} = e_{2}$ (resp. $\hat{e}_{1} = e_{1}$).
We consider the cycle
$\tilde{C} = \lbrace k, \hat{e}_{1}, 2, \hat{e}_{2}, j, e_{j}, j+1, \ldots,
e_{i-1}, i, e_{i}, \ldots, e_{k-1}, k \rbrace$
as represented in Figure~\ref{fig1contractedcycle}.
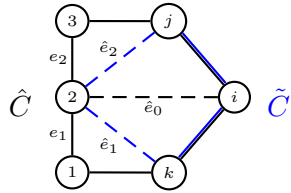
\begin{figure}[!h]
\centering
\begin{pspicture}(0,0)(0,2.5)
\tiny
\begin{psmatrix}[mnode=circle,colsep=0.4,rowsep=0.5]
{$3$} & & {$j$}\\
{$2$} & & & {$i$}\\
{$1$} & & {$k$}
\psset{arrows=-, shortput=nab,labelsep={0.02}}
\tiny
\ncline{1,1}{1,3}
\ncline[linecolor=black,shadow=true,shadowsize=1pt,shadowangle=45,shadowcolor=blue]{1,3}{2,4}
\ncline{1,1}{2,1}_{$e_{2}$}
\ncline{2,1}{3,1}_{$e_{1}$}
\ncline{3,1}{3,3}
\ncline[linecolor=black,shadow=true,shadowsize=1pt,shadowangle=135,shadowcolor=blue]{3,3}{2,4}
\ncline[linestyle=dashed]{2,1}{2,4}_{$\hat{e}_{0}$}
\ncline[linestyle=dashed,linecolor=blue]{2,1}{1,3}^{$\hat{e}_{2}$}
\ncline[linestyle=dashed,linecolor=blue]{2,1}{3,3}_{$\hat{e}_{1}$}
\normalsize
\uput[0](1.2,1){\textcolor{blue}{$\tilde{C}$}}
\uput[0](-2.2,1){\textcolor{black}{$\hat{C}$}}
\end{psmatrix}
\end{pspicture}
\caption{$\hat{C}$ and $\tilde{C}$.}
\label{fig1contractedcycle}
\end{figure}
By renumbering as follows,
$\tilde{1} = k$, $\tilde{e}_{1} = \hat{e}_{1}$, $\tilde{2} = 2$,
$\tilde{e}_{2} = \hat{e}_{2}$, $\tilde{3} = j$,
$\tilde{e}_{\tilde{i} - 1} = e_{i-1}$, $\tilde{i} = i$, $\tilde{e}_{\tilde{i}} = e_{i}$,
$\ldots$,
$\tilde{e}_{m} = e_{k-1}$, $k = \tilde{1}$
we get
$\tilde{C} = \lbrace \tilde{1}, \tilde{e}_{1}, \tilde{2}, \tilde{e}_{2}, \tilde{3}, \ldots,
\tilde{e}_{\tilde{i}-1}, \tilde{i}, \tilde{e}_{\tilde{i}}, \ldots, \tilde{e}_{m}, \tilde{1} \rbrace$.
$\tilde{C}$ has only 
one non-maximum weight chord
$\tilde{e}_{0} = \lbrace 2, \tilde{i} \rbrace$
corresponding to
$\hat{e}_{0} = \lbrace 2, i \rbrace$
and $\tilde{e}_{1}$, $\tilde{e}_{2}$
are the two non-maximum weight edges of $\tilde{C}$.
The restricted path
$\tilde{\gamma} = \lbrace \tilde{3}, \tilde{e}_{3}, \ldots,
\tilde{e}_{\tilde{i}-1}, \tilde{i}, \tilde{e}_{\tilde{i}}, \ldots, \tilde{e}_{m}, \tilde{1} \rbrace$
is a shortest path linking $\tilde{3}$ to $\tilde{1}$
in $\tilde{G}$
as $\gamma$ is a shortest path linking $1$ to $3$ in $\tilde{G}$.
As $C$ and $C^{'}$ have no common chord,
$\tilde{e}_{0} = \lbrace 2, \tilde{i} \rbrace$ cannot be a common chord of
$C$ and $C^{'}$.
We can assume w.l.o.g.
that $\tilde{e}_{0} = \lbrace 2, \tilde{i} \rbrace$ is not a chord of $C^{'}$.
Let $\gamma_{1}$ (resp. $\gamma_{2}$)
be the part of $\gamma$ linking $1$ and $k$ (resp. $3$ and $j$),
$\gamma_{1}= \lbrace k, e_{k}, k+1, \ldots, p, e_{p}, 1 \rbrace$
(resp. $\gamma_{2} = \lbrace 3, e_{3}, 4, \ldots, j-1, e_{j-1}, j \rbrace$)
then
$\gamma_{1} \cup (C^{'} \setminus \lbrace 2 \rbrace) \cup \gamma_{2}$
corresponds to a path linking $k$ to $j$
and therefore $\tilde{1}$ to $\tilde{3}$
in $\tilde{G} \setminus \lbrace 2, \tilde{i} \rbrace$
as represented in Figure~\ref{fig2contractedcycle}.
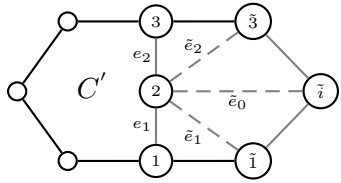
\begin{figure}[!h]
\centering
\begin{pspicture}(0,0)(0,2.2)
\tiny
\begin{psmatrix}[mnode=circle,colsep=0.4,rowsep=0.4]
	& {}	&  	&  {$3$} & & {$\tilde{3}$}\\
{} 	& 	&	&  {$2$} & & & {$\tilde{i}$}\\
	& {} 	& 	& {$1$} & &  {$\tilde{1}$}
\psset{arrows=-, shortput=nab,labelsep={0.02}}
\tiny
\ncline[linestyle=dashed,linecolor=gray]{2,4}{1,6}^{$\tilde{e}_{2}$}
\ncline[linestyle=dashed,linecolor=gray]{2,4}{2,7}_{$\tilde{e}_{0}$}
\ncline[linestyle=dashed,linecolor=gray]{2,4}{3,6}_{$\tilde{e}_{1}$}
\ncline{2,1}{3,2}
\ncline{3,2}{3,4}
\ncline{2,1}{1,2}
\ncline[linecolor=gray]{2,4}{3,4}_{$e_{1}$}
\ncline[linecolor=gray]{2,4}{1,4}^{$e_{2}$}
\ncline{1,2}{1,4}
\ncline{1,4}{1,6}
\ncline[linecolor=gray]{1,6}{2,7}
\ncline[linecolor=gray]{2,7}{3,6}
\ncline{3,6}{3,4}
\normalsize
\uput[0](-2.4,1.1){\textcolor{black}{$C^{'}$}}
\end{psmatrix}
\end{pspicture}
\caption{Path $\lbrace \tilde{1}, 1 \rbrace \cup (C^{'} \setminus \lbrace 2 \rbrace) \cup \lbrace 3, \tilde{3} \rbrace$.}
\label{fig2contractedcycle}
\end{figure}
Therefore there exists a shortest path $\gamma^{'}$
in $\tilde{G} \setminus \lbrace 2, \tilde{i} \rbrace$ linking $\tilde{1}$
to $\tilde{3}$.
Then we can consider the cycles $\tilde{C}$ and 
$\tilde{C}^{'} = \lbrace 2, \tilde{e}_{1}, \tilde{1}\rbrace
\cup \gamma^{'} \cup \lbrace \tilde{3}, \tilde{e}_{2}, 2 \rbrace$
as represented in Figure~\ref{fig3contractedcycle}.
\begin{figure}[!h]
\centering
\begin{pspicture}(0,0)(0,2.2)
\tiny
\begin{psmatrix}[mnode=circle,colsep=0.4,rowsep=0.4]
	& {}	&  	&  {} & & {$\tilde{3}$}\\
{} 	& 	&	&  {$2$} & & & {$\tilde{i}$}\\
	& {} 	& 	& {} & &  {$\tilde{1}$}
\psset{arrows=-, shortput=nab,labelsep={0.02}}
\tiny
\ncline[linecolor=black,shadow=true,shadowsize=1pt,shadowangle=-45,shadowcolor=blue]{1,6}{2,4}_{$\tilde{e}_{2}$}
\ncline[linestyle=dashed]{2,4}{2,7}_{$\tilde{e}_{0}$}
\ncline[linecolor=black,shadow=true,shadowsize=1pt,shadowangle=45,shadowcolor=blue]{2,4}{3,6}_{$\tilde{e}_{1}$}
\ncline{2,1}{3,2}
\ncline{3,2}{3,4}
\ncline{2,1}{1,2}
\ncline{1,2}{1,4}
\ncline{1,4}{1,6}
\ncline[linecolor=blue]{1,6}{2,7}
\ncline[linecolor=blue]{2,7}{3,6}
\ncline{3,6}{3,4}
\normalsize
\uput[0](1.2,1.1){\textcolor{blue}{$\tilde{C}$}}
\uput[0](-3.8,1.1){\textcolor{black}{$\tilde{C}^{'}$}}
\end{psmatrix}
\end{pspicture}
\caption{Adjacent cycles $\tilde{C}$ and $\tilde{C}^{'}$.}
\label{fig3contractedcycle}
\end{figure}
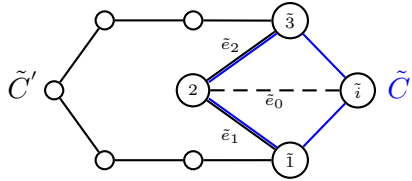
We can achieve the proof as in the first case.
By construction
$\tilde{C}$ has only one non-maximum weight chord $\tilde{e}_{0} = \lbrace 2, \tilde{i} \rbrace$
which is not a chord of $\tilde{C}^{'}$ ($\tilde{i} \notin \tilde{C}^{'}$)
and $\tilde{C}$ and $\tilde{C}^{'}$ have no maximum weight chords.
$\tilde{C}$ and $\tilde{C}^{'}$  satisfy conditions
\ref{enumPropV(C)-V(C)^{'}nonempty} to \ref{enumPropCandC'HaveNoCommonChord},
\ref{LemItemItemCC{'}twocommonmaxweightedges2} and \ref{LemItemItemC2shortestpathG213iV13C'2G2i13}
(replacing $e_{1}$ (resp. $e_{2}$) by $\tilde{e}_{1}$ (resp. $\tilde{e}_{2}$)).

As in Case 1,
we can obtain $\hat{C}$,
if it exists,
using a Breadth First Search algorithm in $O(n^{2})$ time.
Then to find a chord $\tilde{e}_{0} = \lbrace 2, \tilde{i} \rbrace$ of $\hat{C}$,
we need at most $O(n)$ time.
To select the edges $\tilde{e}_{1} = \lbrace 2, k \rbrace$
and $\tilde{e}_{2} = \lbrace 2, j \rbrace$,
we also need at most $O(n)$ comparisons.
To build the shortest path $\gamma^{'}$ in $\tilde{G} \setminus \lbrace 2, \tilde{i} \rbrace$
using a Breadth First Search algorithm,
we need $O(n^{2})$ time.
Hence to verify Claim~\ref{LemItem2cyclesCC'2}
when $\hat{C}$ has a chord,
we need $O(n^{2})$ time.
Note that if there exists a chord in $\tilde{C}$
we do not need to check Claim~\ref{LemItemItemC2shortestpathG213iV13C'2G2i13}
for all $i \in V(\tilde{\gamma}) \setminus \lbrace \tilde{1}, \tilde{3} \rbrace$.
We check the condition only for the node $i$ corresponding to the chord we have found.
\end{proof}

We consider that the Adjacent cycles condition can be divided in two parts.
The first part is the condition of non-existence of cycles with two common non-maximum weight edges.
Then the second part is the condition on edge-weights
for cycles with a unique common non-maximum weight edge.

\begin{proposition}
\label{PropFirstPartAdjacentCyclesCondPolyTime}
The first part of the Adjacent cycles condition can be verified
in $O(n^{6})$ time.
\end{proposition}

\begin{proof}
We have to verify that
for any pair of adjacent edges $e_{1} = \lbrace 1,2 \rbrace$ and $ e_{2} = \lbrace 2,3 \rbrace$,
$e_{1}$ and $e_{2}$ are not common non-maximum weight edges of two adjacent cycles
(\emph{i.e.}, Claim~\ref{LemItem2cyclesCC'} of Lemma~\ref{lemMINADJCYCLESCOND} is not satisfied).
By Lemma~\ref{lemMINADJCYCLESCOND}
it is sufficient to prove that Claim~\ref{LemItem2cyclesCC'2}
of Lemma~\ref{lemMINADJCYCLESCOND} is not satisfied
and it can be done
in $O(n^{3})$ time.
There are at most
$\sum_{i \in V} C_{n}^{2} = n^{2} \left(\frac{n-1}{2}\right)$
pairs of adjacent edges.
Hence
we need at most $n^{3}$
iterations of the procedure considered in Lemma~\ref{lemMINADJCYCLESCOND}
and therefore the first part of the Adjacent cycles condition can be verified
in $O(n^{6})$ time.
\end{proof}

\begin{lemma}
\label{lemMinCyclesOneCommonNonMaxWeightEdge}
Let $G = (N,E,w)$ be a weighted graph satisfying the Star and Path
(and therefore Intermediary Cycle)
conditions.
Let us assume that the first part of the Adjacent cycles condition is satisfied.
Then the following Claims are equivalent:
\begin{enumerate}[1)]
\item \label{LemItem2cyclesCC'3}
There exist two simple cycles $C$ and $C^{'}$ in $G$
satisfying conditions
\ref{enumPropV(C)-V(C)^{'}nonempty},
\ref{enumPropCatmost1non-maxweightchord},
\ref{enumPropCC^{'}nomaxweightchord},
\ref{enumPropCandC'HaveNoCommonChord}
of the Adjacent Cycles condition,
and such that:
\begin{enumerate}
\setcounter{enumii}{4}
\item
\label{LemItemItemCC^{'}twocommonnonmaxweightedges2}
$C$ and $C^{'}$ have a unique common non-maximum weight edge
$e_{1} = \lbrace 1,2 \rbrace$
and $w(e) > w_{1}$
for all $e \in (E(C) \cup E(C^{'})) \setminus \lbrace e_{1} \rbrace$.
\end{enumerate}
\item
\label{LemItem2cyclesCC'4}
There exist two simple cycles $\tilde{C}$ and $\tilde{C}^{'}$ in $G$
satisfying conditions
\ref{enumPropV(C)-V(C)^{'}nonempty},
\ref{enumPropCatmost1non-maxweightchord},
\ref{enumPropCC^{'}nomaxweightchord},
\ref{enumPropCandC'HaveNoCommonChord}
of the Adjacent Cycles condition,
and such that:
\begin{enumerate}
\setcounter{enumii}{5}
\item
\label{LemItemItemCC{'}twocommonmaxweightedges22}
$\tilde{C}$ and $\tilde{C}^{'}$ have a unique common non-maximum weight edge
$\tilde{e}_{1} = \lbrace 1, 2 \rbrace$
and $w(e) > w_{1}$
for all $e \in (E(\tilde{C}) \cup E(\tilde{C}^{'})) \setminus \lbrace e_{1} \rbrace$.
\item
\label{LemItemItemC2shortestpathG213iV13C'2G2i132}
$\tilde{C} \setminus \lbrace \tilde{e}_{1} \rbrace$
corresponds to a shortest path $\tilde{\gamma}$
in $\tilde{G} = (V, E \setminus \lbrace \tilde{e}_{1} \rbrace,
w_{| E \setminus \lbrace \tilde{e}_{1} \rbrace}
)$
linking $1$ and $2$
and for some $i \in V(\tilde{\gamma}) \setminus \lbrace 1, 2 \rbrace$,
$\tilde{C}^{'} \setminus \lbrace \tilde{e}_{1} \rbrace$
corresponds to a shortest path $\tilde{\gamma}^{'}$
in $\tilde{G}_{N \setminus \lbrace i \rbrace}$ linking $1$ and $2$.
\end{enumerate}
\end{enumerate}
\end{lemma}

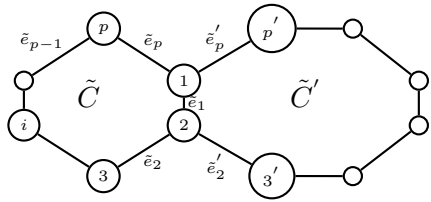
\begin{figure}[!h]
\centering
\begin{pspicture}(0,-.2)(0,2.4)
\tiny
\begin{psmatrix}[mnode=circle,colsep=0.6,rowsep=0.1]
				& {$p$}	&  			& {$p^{'}$} & { }\\
{} 			& 	 		& {$1$}	&   				& 		&  { }\\
{$i$} 	& 	 		& {$2$}	&    				& 		& { }\\
				& {$3$} & 			& {$3^{'}$} 			& { }
\psset{arrows=-, shortput=nab,labelsep={0.02}}
\tiny
\ncline{2,1}{1,2}^{$\tilde{e}_{p-1}$}
\ncline{1,2}{2,3}^{$\tilde{e}_{p}$}
\ncline{2,1}{3,1}
\ncline{3,1}{4,2}
\ncline{4,2}{3,3}_{$\tilde{e}_{2}$}
\ncline{2,3}{3,3}^{$\tilde{e}_{1}$}
\ncline{2,3}{1,4}^{$\tilde{e}_{p}^{'}$}
\ncline{1,4}{1,5}
\ncline{1,5}{2,6}
\ncline{2,6}{3,6}
\ncline{3,3}{4,4}_{$\tilde{e}_{2}^{'}$}
\ncline{4,4}{4,5}
\ncline{4,5}{3,6}
\end{psmatrix}
\normalsize
\uput[0](-4.8,1.1){$\tilde{C}$}
\uput[0](-2,1.1){$\tilde{C}^{'}$}
\end{pspicture}
\caption{$\tilde{C} = \lbrace 1, \tilde{e}_{1}, 2, \tilde{e}_{2}, \ldots,
p, \tilde{e}_{p}, 1 \rbrace$
and $\tilde{C}^{'} = \lbrace  1, \tilde{e}_{1}, 2, \tilde{e}_{2}^{'}, \ldots,
p^{'}, \tilde{e}_{p^{'}}^{'}, 1 \rbrace $.
}
\label{figComplexityAdjacentCyclesWithOneCommonedge}
\end{figure}

Let us observe that using the Intermediary Cycle condition
we only need to verify
that the edges in $E(C)$ adjacent to $e_{1}$ have a weight
strictly greater than $w_{1}$.
Indeed if it is satisfied
then $w(e) > w_{1}$ for all $e \in E(C) \setminus \lbrace e_{1} \rbrace$.

\begin{remark}
Let us note that by the Intermediary Cycle condition
$C, C^{'}, \tilde{C}, \tilde{C}^{'}$ can have at most two non-maximum weight edges.
If they all have two non-maximum weight edges then $e_{1}$
and all these edges are adjacent (incident to $1$ or $2$).
\end{remark}

\begin{proof}
Let us consider $C = \lbrace 1, e_{1}, 2, e_{2}, \ldots, m, e_{m}, 1 \rbrace$.
Then
$\gamma = C \setminus \lbrace e_{1} \rbrace =
\lbrace 2, e_{2}, \ldots, m, e_{m}, 1 \rbrace$
is a path linking $1$ and $2$ in $\tilde{G}$.
Hence there is a shortest path
$\tilde{\gamma} =
\lbrace 2, \tilde{e}_{2}, \tilde{3}, \tilde{e}_{3}, \ldots, p, \tilde{e}_{p}, 1 \rbrace$
linking $1$ and $2$ in $\tilde{G}$.
Then we can build the cycle
$\tilde{C} :=
\lbrace 1, e_{1}, 2, \tilde{e}_{2}, \tilde{3}, \tilde{e}_{3}, \ldots, p, \tilde{e}_{p}, 1 \rbrace$
adding the edge $e_{1}$ to the path $\tilde{\gamma}$.
Either $\tilde{C} = C$,
or $\tilde{C}$ and $C$ are adjacent.
In this last case
Lemma~\ref{lemM=maxu(e)=maxu(e)=M'} implies
$\max_{e \in E(C)} w(e) = \max_{e \in E(\tilde{C})} w(e)$.
Hence $e_{1}$ is a non-maximum weight edge in $\tilde{C}$.
Claim~\ref{LemItemItemCC^{'}twocommonnonmaxweightedges2}
implies $w_{2} > w_{1}$
(resp. $w_{m} > w_{1}$).
If $\tilde{e}_{2} = e_{2}$
(resp $\tilde{e}_{p} = e_{m}$)
then $w(\tilde{e}_{2}) = w_{2}$
(resp. $w(\tilde{e}_{p}) = w_{m}$).
Otherwise Star condition applied to $\lbrace e_{1}, e_{2}, \tilde{e}_{2} \rbrace$
(resp. $e_{1}, e_{m}, \tilde{e}_{p}$)
implies $w(\tilde{e}_{2}) = w_{2}$
(resp. $w(\tilde{e}_{p}) = w_{m}$).
Then the Intermediary Cycle condition
implies
$w_{1} < w(\tilde{e}_{2}) \leq w(\tilde{e}_{3}) = \ldots = w(\tilde{e}_{p}) = M$
with $M = \max_{e \in E(\tilde{C})} w(e)$,
exchanging $\tilde{e}_{2}$ and $\tilde{e}_{p}$ if necessary.
The Intermediary Cycle condition also implies
that for any chord $e$ in $\tilde{C}$,
we have $w(e) \leq w_{2}$
if $e$ is incident to $2$
(in fact, as $w_{1}< w_{2}$,
the Star condition applied to $\lbrace e_{1}, \tilde{e}_{2}, e \rbrace$
implies $w(e) = w_{2}$
)
and $w(e) = M$ otherwise.
In any case $e$ would contradict the optimality of $\tilde{\gamma}$.
Hence $\tilde{C}$ has no chord.

Let us assume $V(\tilde{C}) \subseteq V(C) \cap V(C^{'})$.
As $V(C) \setminus V(C^{'}) \not= \emptyset$
(resp. $V(C^{'}) \setminus V(C) \not= \emptyset$)
we have $V(\tilde{C}) \not= V(C)$
(resp. $V(\tilde{C}) \not= V(C^{'})$).
Then at least one edge $\tilde{e}$ of $\tilde{C}$ is a chord of $C$.
As $C$ has no maximum weight chord,
$\tilde{e}$ is a chord of non-maximum weight in $C$
and then $\tilde{e}$ is also an edge in $E(\tilde{C}) \setminus \lbrace e_{1} \rbrace$
of non-maximum weight
(by Lemma~\ref{lemM=maxu(e)=maxu(e)=M'}, 
as $C$ and $\tilde{C}$ are adjacent the maximum weight in $C$ and $\tilde{C}$ is $M$).
Following the Intermediary Cycle condition
$\tilde{e}$ is the unique edge of $\tilde{C}$ of non-maximum weight
adjacent to $e_{1}$
as represented in Figure~\ref{figComplexityAdjacentCyclesWithOneCommonedge2}.
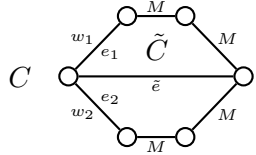
\begin{figure}[!h]
\centering
\begin{pspicture}(0,-.1)(0,2)
\tiny
\begin{psmatrix}[mnode=circle,colsep=0.5,rowsep=0.5]
	 & {}	& {}\\
{} & 	 	& 	 & {}\\
	 & {} & {} 
\psset{arrows=-, shortput=nab,labelsep={0.02}}
\tiny
\ncline{2,1}{1,2}_{$e_{1}$}^{$w_{1}$}
\ncline{2,1}{3,2}^{$e_{2}$}_{$w_{2}$}
\ncline{1,2}{1,3}^{$M$}
\ncline{1,3}{2,4}^{$M$}
\ncline{3,2}{3,3}_{$M$}
\ncline{3,3}{2,4}_{$M$}
\ncline{2,1}{2,4}_{$\tilde{e}$}
\end{psmatrix}
\normalsize
\uput[0](-3.4,.8){$C$}
\uput[0](-1.6,1.2){$\tilde{C}$}
\end{pspicture}
\caption{$\tilde{e}$ non-maximum weight chord (resp. edge) in $C$ (resp. $\tilde{C}$).
}
\label{figComplexityAdjacentCyclesWithOneCommonedge2}
\end{figure}
By the same reasoning
(interchanging $C$ and $C^{'}$),
there exists an edge $\tilde{e}^{'}$ of $\tilde{C}$
which is a chord of $C^{'}$
and $\tilde{e}^{'}$ is an edge of $\tilde{C}$ of non-maximum weight.
$\tilde{e}^{'}$ necessarily corresponds to the unique edge of $\tilde{C}$ of non-maximum weight
adjacent to $e_{1}$.
Therefore $\tilde{e} = \tilde{e}^{'}$.
Hence $C$ and $C^{'}$ have a common chord, a contradiction.

Therefore we have
either $V(\tilde{C}) \setminus V(C) \not= \emptyset$
or $V(\tilde{C}) \setminus V(C^{'}) \not= \emptyset$.
In every case we can choose $i \in V(\tilde{C}) \setminus \lbrace 1,2 \rbrace$
such that there exists a path $\hat{\gamma}$
in $\tilde{G}_{N \setminus \lbrace i \rbrace}$
linking $1$ and $2$.
$\hat{\gamma} = \gamma$ if $i \notin V(C)$
or
$\hat{\gamma} =
C^{'} \setminus \lbrace e_{1} \rbrace =
\lbrace 2, e_{2}^{'}, 3^{'}, e_{3}^{'}, \ldots, m^{'}, e_{m^{'}}^{'}, 1 \rbrace$
if $i \notin V(C^{'})$
as represented in Figure~\ref{fig2contractedcycle2}.
Therefore we can find a shortest path $\tilde{\gamma}^{'}$
in $\tilde{G}_{N \setminus \lbrace i \rbrace}$
linking $1$ and $2$
as represented in Figure~\ref{figComplexityAdjacentCyclesWithOneCommonedge}.
Adding to $\tilde{\gamma}^{'}$ the edge $e_{1}$
we get the cycle
$\tilde{C}^{'}
= \lbrace 1, e_{1}, 2, \tilde{e}_{2}^{'}, \tilde{3}^{'}, \tilde{e}_{\tilde{3}^{'}}^{'}, \ldots,
\tilde{p}^{'}, \tilde{e}_{\tilde{p}^{'}}^{'}, 1 \rbrace$.
$\tilde{C}^{'}$ has no chord
otherwise $\tilde{\gamma}^{'}$ is not a shortest path
in $\tilde{G}_{N \setminus \lbrace i \rbrace}$
(by the same reasoning as for $\tilde{C}$).
As $i \in V(\tilde{C}) \setminus V(\tilde{C}^{'})$
we have $V(\tilde{C}) \setminus V(\tilde{C}^{'}) \not= \emptyset$.
We have already seen that
$\max_{e \in E(C)} w(e) = \max_{e \in E(\tilde{C})} w(e)
= \max_{e \in E(\tilde{C}^{'})} w(e)$,
and that if $C$
contains a second non-maximum weight edge
then $\tilde{C}$ and $\tilde{C^{'}}$ also contain
a second non-maximum weight edge of same weight.
Therefore
if $V(\tilde{C}^{'}) \subset V(\tilde{C})$
then $\tilde{\gamma}^{'}$ is a path linking $1$ and $2$ in $\tilde{G}$
shorter that $\tilde{\gamma}$, a contradiction.
Therefore we have $V(\tilde{C}^{'}) \setminus V(\tilde{C}) \not= \emptyset$.

\begin{figure}[!h]
\centering
\begin{pspicture}(0,0)(0,2.4)
\tiny
\begin{psmatrix}[mnode=circle,colsep=0.4,rowsep=0.4]
	& {$3^{'}$}	&  	&  {$2$} & & {$3$}\\
{} 	& 	&	&  {$1$} & & & {$\tilde{3}$}\\
	& {} 	& 	& {$m^{'}$} & &  {$i$}
\psset{arrows=-, shortput=nab,labelsep={0.02}}
\tiny
\ncline[linestyle=dotted]{1,4}{1,6}
\ncline[linestyle=dotted]{1,6}{2,7}
\ncline{1,4}{2,7}^{$\tilde{e}_{2}$}
\ncline[linestyle=dashed,linecolor=gray]{1,4}{3,6}_{$\tilde{e}_{0}$}
\ncline{2,1}{3,2}
\ncline{3,2}{3,4}
\ncline{2,1}{1,2}
\ncline[linecolor=gray]{2,4}{3,4}
\ncline[linecolor=gray]{2,4}{1,4}^{$e_{1}$}
\ncline{1,2}{1,4}^{$e_{2}^{'}$}
\ncline[linecolor=gray]{2,7}{3,6}
\ncline{3,6}{3,4}
\normalsize
\uput[0](-2.4,1.1){\textcolor{black}{$C^{'}$}}
\uput[0](-0.3,1.1){\textcolor{black}{$\tilde{C}$}}
\end{psmatrix}
\end{pspicture}
\caption{$C^{'}$ and $\tilde{C}$, with $i \notin C^{'}$.}
\label{fig2contractedcycle2}
\end{figure}
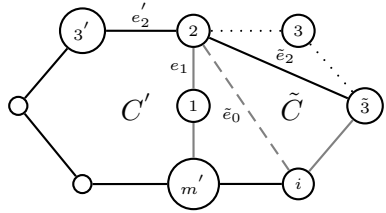

\end{proof}

We first verify a weaker condition than the second part of the Adjacent Cycles condition.\\

\noindent 
\framebox[\linewidth]{
\parbox{.95\linewidth}{
\textbf{
Weak Second Part of the Adjacent Cycles Condition.
}
If $e_{1}$ is a unique non-maximum weight edge common
to two simple cycles $(C, C^{'})$ in $G$
satisfying the conditions
\ref{enumPropV(C)-V(C)^{'}nonempty},
\ref{enumPropCatmost1non-maxweightchord},
\ref{enumPropCC^{'}nomaxweightchord},
\ref{enumPropCandC'HaveNoCommonChord}
of the Adjacent Cycles condition, 
then
there are non-maximum weight edges $e_{2} \in E(C) \setminus \lbrace e_{1} \rbrace$
and $e_{2}^{'} \in E(C^{'}) \setminus \lbrace e_{1} \rbrace$
such that $e_{1}, e_{2}, e_{2}^{'}$ are adjacent
and $w_{1} = w_{2} \geq w_{2}^{'}$
or $w_{1} = w_{2}^{'} \geq w_{2}$.
}
}

\vspace{.4cm}
Note that we cannot have $e_{2} = e_{2}^{'}$
otherwise $e_{1}$ and $e_{2}$ are two non-maximum weight edges common to $C$ and $C^{'}$
as $w_{1} = w_{2} = w_{2}^{'}$.
Therefore we have $e_{2} \in E(C) \setminus E(C^{'})$
and $e_{2}^{'} \in E(C^{'}) \setminus E(C)$.

\begin{proposition}
Let $G = (N,E,w)$ be a weighted graph satisfying the Star and Path
(and therefore Weak and Intermediary Cycle)
conditions,
and the first part of the Adjacent Cycles condition.
Then the Weak Second Part of the Adjacent Cycles condition can be verified
in $O(n^{3}m)$ time.
\end{proposition}

\begin{proof}
Let $e_{1}$ be a common non-maximum weight edge
of two adjacent cycles $C$ and $C^{'}$
satisfying the conditions
\ref{enumPropV(C)-V(C)^{'}nonempty},
\ref{enumPropCatmost1non-maxweightchord},
\ref{enumPropCC^{'}nomaxweightchord},
\ref{enumPropCandC'HaveNoCommonChord}
of the Adjacent Cycles condition.
Let $e_{2} \in E(C) \setminus e_{1}$ and $e_{2}^{'} \in E(C^{'}) \setminus e_{1}$
be edges incident to the same end-vertex of $e_{1}$
(at this step of the proof
$e_{2}$ may be equal to $e_{2}^{'}$
and $e_{2}$ or $e_{2}^{'}$ may have maximum weight in $C$ or $C^{'}$).
Let us assume
that we have
neither $w_{1} = w_{2} \geq w_{2}^{'}$
nor $w_{1} = w_{2}^{'} \geq w_{2}$.
If $e_{2} = e_{2}^{'}$
then $w_{2} = w_{2}^{'}$
and therefore we have $w_{1} < w_{2} = w_{2}^{'}$
or $w_{1} > w_{2} = w_{2}^{'}$.
This last case is not possible as $e_{1}$, $e_{2}$
would be two non-maximum weight edges common to $C$ and $C^{'}$.
If $e_{2} \not= e_{2}^{'}$
then Star condition applied to $e_{1}, e_{2}, e_{2}^{'}$
also implies $w_{1} < w_{2} = w_{2}^{'}$.
(By the same reasoning we get $w_{1} < w_{m} = w_{m^{'}}^{'}$.)
Then the Intermediary Cycle condition implies
$w(e) > w_{1}$ for all $e \in (E(C) \cup E(C^{'})) \setminus \lbrace e_{1}  \rbrace$.
Then by Lemma~\ref{lemMinCyclesOneCommonNonMaxWeightEdge},
$e_{1}$ is also a non-maximum weight edge
common to two cycles $\tilde{C}$ and $\tilde{C}^{'}$
such that $\tilde{C} \setminus \lbrace e_{1} \rbrace$
corresponds to a shortest path $\tilde{\gamma}$
linking $1$ and $2$ in
$\tilde{G} =
(N, E \setminus \lbrace e_{1} \rbrace, w_{|E \setminus \lbrace e_{1} \rbrace})$
and for some $i \in V(\tilde{\gamma}) \setminus \lbrace 1, 2 \rbrace$,
$\tilde{C}^{'} \setminus \lbrace e_{1} \rbrace$
corresponds to a shortest path in $\tilde{G}_{N \setminus \lbrace i \rbrace}$
linking $1$ and $2$.
Moreover
we still have $w(e) > w_{1}$
for all $e \in (E(\tilde{C}) \cup E(\tilde{C}^{'})) \setminus \lbrace e_{1} \rbrace$.
Therefore for every edge $e_{1} \in E$,
we only have to verify
that such a couple of adjacent cycles $(\tilde{C}, \tilde{C}^{'})$ cannot exist.
Therefore
we look for a shortest path $\tilde{\gamma}$ linking $1$ and $2$
in $\tilde{G}$.
If it exists then we look for a shortest path $\tilde{\gamma}^{'}$
in $\tilde{G}_{N \setminus \lbrace i \rbrace}$
with $i \in V(\tilde{\gamma}) \setminus \lbrace 1, 2 \rbrace$,
linking $1$ and $2$.
If such a path $\tilde{\gamma}^{'}$ exists, 
we only have to verify that $w(e) = w_{1}$
for some edge $e$ of $\tilde{\gamma}$ or $\tilde{\gamma}^{'}$
(in fact the
Weak Cycle condition implies that such an edge is necessarily
adjacent to $e_{1}$).

Note that by the Star and Weak Cycle conditions
and Lemma~\ref{lemM=maxu(e)=maxu(e)=M'},
all elementary paths linking $1$ and $2$ in $\tilde{G}$
have either all their edges of same weight $M > w_{1}$
or one edge incident to $1$ or $2$ with weight $w_{2}> w_{1}$
and their remaining edges with weight $M > w_{2}$.
Therefore we can find a shortest path $\tilde{\gamma}$
linking $1$ and $2$ in $\tilde{G}$
using a BFS algorithm in $O(n^{2})$.
Then we can also use a BFS algorithm to find a shortest path
in
$\tilde{G}_{N \setminus \lbrace i \rbrace}$
with $i \in \tilde{\gamma}$.
In the worst case
we have to consider all nodes $i \in \tilde{\gamma}$.
Therefore we can verify in $O(n^{3})$ time
if $e_{1}$ satisfies the Weak Second Part of the Adjacent Cycles condition.
Therefore the Weak Second Part of the Adjacent cycles condition can be verified
in $O(n^{3}m)$ time.
\end{proof}

\begin{proposition}
Let $G = (N,E,w)$ be a weighted graph satisfying the Star and Path
(and therefore Intermediary Cycle)
conditions,
and the Weak Second Part of the Adjacent Cycles condition.
Then the Cycle condition is satisfied.
\end{proposition}

\begin{proof}
Let us consider a cycle
$C = \lbrace 1, e_{1}, 2, e_{2}, \ldots, m, e_{m},1 \rbrace$.
The Intermediary Cycle condition implies that
$w_{1} \leq w_{2} \leq w_{3} = \cdots = w_{m} = M$
after renumbering if necessary
and that
$w(e) \leq w_{2}$ for all chord incident to $2$
and $w(e) \leq M$ for all other chords.
Moreover:
\begin{itemize}
\item
If $w_{1} \leq w_{2} < M$
then $w(e) = M$ for all $e \in E(C)$  non-incident to~$2$.
If $e$ is a chord incident to $2$
then $w_{1} \leq w_{2} = w(e) < M$ or $w(e) < w_{1} = w_{2} < M$.
\item
If $w_{1} < w_{2} = M$,
then $w(e) = M$ for all $e \in E(C) \setminus \lbrace e_{1} \rbrace$.
\end{itemize}
Therefore if $w_{1} < w_{2} = M$
the Cycle condition is satisfied.
If $w_{1} \leq w_{2} < M$
let us assume by contradiction that
 there exists a chord incident to $2$
with $w(e) < w_{2}$.
Star condition applied to $\lbrace e, e_{1}, e_{2}\rbrace$
implies $w_{1} = w_{2}$
and that for any other chord $e^{'}$ incident to $2$
we have  $w(e^{'}) = w_{1} = w_{2}$.
Hence we can assume that $C$ has only one chord $e$ incident to $2$
replacing if necessary $C$
by a smaller cycle
using the chords of $C$ incident to $2$
which are as close as possible to $e$.
We have $e = \lbrace 2, j \rbrace$ with $4 \leq j \leq m$.
We can consider two adjacent cycles
$\tilde{C} := \lbrace 2, e, j, e_{j}, \ldots, m, e_{m}, 1, e_{1}, 2 \rbrace$
and
$\tilde{C}^{'} := \lbrace 2, e_{2}, 3, e_{3}, \ldots, e_{j-1}, j, e, 2 \rbrace$
as represented
in Figure~\ref{figStarPathFirstPartAndWeakSecondPartOfAdjacentCyclesImplyCycleCondition}.

\begin{figure}[!h]
\centering
\begin{pspicture}(0,0)(0,2.6)
\tiny
\begin{psmatrix}[mnode=circle,colsep=0.7,rowsep=0.2]
				& {$j$}\\
{$m$} 			& 	 		& {}\\
{$1$} 	& 	 		& {$3$}\\
				& {$2$}
\psset{arrows=-, shortput=nab,labelsep={0.02}}
\tiny
\ncline{2,1}{1,2}^{$M$}
\ncline{1,2}{2,3}^{$M$}
\ncline{2,1}{3,1}_{$M$}
\ncline{3,1}{4,2}^{$e_{1}$}_{$w_{1}$}
\ncline{4,2}{3,3}^{$e_{2}$}_{$w_{2}$}
\ncline{2,3}{3,3}^{$M$}
\ncline{1,2}{4,2}^{$e$}
\end{psmatrix}
\normalsize
\uput[0](-2.4,1.3){$\tilde{C}$}
\uput[0](-1.2,1.3){$\tilde{C}^{'}$}
\end{pspicture}
\caption{$w(e) < w_{1} = w_{2}$.}
\label{figStarPathFirstPartAndWeakSecondPartOfAdjacentCyclesImplyCycleCondition}
\end{figure}

\noindent
$e$ is a unique common edge to $\tilde{C}$ and $\tilde{C}^{'}$
such that $w(e) < w_{1} = w_{2} \leq M$.
If $\tilde{C}$
(resp. $\tilde{C}^{'}$) has a chord
then we can replace $\tilde{C}$
(resp. $\tilde{C}^{'}$)
by a smaller cycle.
Therefore we can suppose that $\tilde{C}$ and $\tilde{C}^{'}$ have no chords.
Then
$w(e) < w_{1} = w_{2}$
contradicts the Weak Second Part of the Adjacent Cycles condition.
Therefore any edge incident to $2$ has weight $w_{2}$.

Finally if $w_{1} = w_{2} = M$,
let us also assume by contradiction that
there exists a chord $e$ in $C$ with $w(e) < w_{2} = M$.
As $w_{1} = w_{2} = M$ we can always assume $e$ incident to $2$ after renumbering if necessary
and then we can establish the same contradiction as before. 
\end{proof}

Assuming now that the Star, Path, and Cycle conditions are satisfied
we give a new formulation of the Pan condition.
Then we will analyze its complexity.

\begin{lemma}
\label{lemPanConditionWithAssumptionStarPathCycleConditionsSatisfied}
Let $G = (N, E, w)$ be a connected weighted graph
satisfying the Star, Path, and Cycle conditions.
The Pan condition is satisfied if and only if
for all simple cycle
$C = \lbrace 1, e_{1}, 2, e_{2}, \ldots, m, e_{m}, 1 \rbrace$,
with
$w_{1} = w_{2} \leq w_{3} = \cdots = w_{m} = \hat{M} = \max_{e \in \hat{E}(C)} w(e)$,
one of the following claim is satisfied:
\begin{enumerate}[1)]
\item
\label{lemNewFormulationOfThePanConditionw1=w2=M}
$w_{1} = w_{2} = M$.
\item
\label{lemNewFormulationOfThePanConditionw1=w2<MAnd13IsAChordOfC}
$w_{1} = w_{2} < M$ and $\lbrace 1, 3 \rbrace$ is a chord of $C$.
\item 
\label{lemNewFormulationOfThePanConditionSigma(E)=w1=w2<MWhereSigma(E)=Min_Ew(e)}
$\sigma(E) = w_{1} = w_{2} < M$, where $\sigma(E) = \min_{e \in E} w(e)$.
\end{enumerate}
\end{lemma}

\begin{proof}
Let us first assume that the Pan condition is satisfied.
Let us consider a simple cycle
$C = \lbrace 1, e_{1}, 2, e_{2}, \ldots, m, e_{m}, 1 \rbrace$,
with
$w_{1} = w_{2} \leq w_{3} = \cdots = w_{m} = \hat{M}$.
If $w_{1} = w_{2} = \hat{M}$
then Claim~\ref{lemNewFormulationOfThePanConditionw1=w2=M}
is satisfied.
Otherwise we have $w_{1} = w_{2} < \hat{M}$.
If $\lbrace 1, 3 \rbrace$ is a chord of $C$
then Claim~\ref{lemNewFormulationOfThePanConditionw1=w2<MAnd13IsAChordOfC} is satisfied.
Otherwise 
$\lbrace 1, 3 \rbrace$ is not a chord of $C$
and let us assume
$\sigma(E) < w_{1} = w_{2}$.
Then there exists $e \in E \setminus E(C)$
such that $w(e) < w_{1} = w_{2}$.
Note that the Cycle condition implies that a chord of $C$
incident (resp. non incident) to $2$ has weight $w_{2}$ (resp. $\hat{M}$).
Therefore we actually have $e \in E \setminus \hat{E}(C)$.
As $G$ is connected
there is a path $P$ containing $e$
($P$ can be restricted to $e$)
such that $|V(C) \cap V(P)| = 1$.
Then as $w(e) < w_{1} = w_{2} < M$,
Pan condition implies
that $V(C) \cap V(P) = \lbrace 2 \rbrace$
and that $\lbrace 1, 3 \rbrace$ is a maximum weight chord of $C$,
a contradiction.
Therefore
Claim~\ref{lemNewFormulationOfThePanConditionSigma(E)=w1=w2<MWhereSigma(E)=Min_Ew(e)}
is satisfied.

Let us now consider a simple cycle
$C = \lbrace 1, e_{1}, 2, e_{2}, \ldots, m, e_{m}, 1 \rbrace$,
and an elementary path $P$ such that there is an edge $e \in P$
with $w(e) \leq \min_{1 \leq k \leq m} w_{k}$
and $|V(C) \cap V(P)| = 1$.
The Cycle condition
implies $w_{1} \leq w_{2} \leq w_{3} = \cdots = w_{m} = \hat{M}$
after renumbering if necessary.
Moreover the Path and Star conditions imply
$w_{1} = w_{2} \leq \hat{M}$
and if $w_{1}= w_{2} < \hat{M}$
then $V(C) \cap V(P) = \lbrace 2 \rbrace$
(as the Weak Pan condition is satisfied by Proposition~\ref{PropPanCondition}).
If $C$ satisfies Claim~\ref{lemNewFormulationOfThePanConditionw1=w2=M}
then the Pan condition is trivially satisfied.
If $C$ satisfies Claim~\ref{lemNewFormulationOfThePanConditionw1=w2<MAnd13IsAChordOfC}
then
$\lbrace 1, 3 \rbrace$ is a chord of $C$
and as $w_{1} = w_{2} < \hat{M}$
we have $V(C) \cap V(P) = \lbrace 2 \rbrace$.
Moreover the (Intermediary) Cycle condition implies
that $\lbrace 1, 3 \rbrace$ is a maximum weight chord.
Therefore the Pan condition is satisfied.
Finally if $C$ satisfies
Claim~\ref{lemNewFormulationOfThePanConditionSigma(E)=w1=w2<MWhereSigma(E)=Min_Ew(e)},
we also have $V(C) \cap V(P) = \lbrace 2 \rbrace$
and as $\sigma(E) = w_{1}$,
we necessarily have $w(e) = w_{1}$.
Therefore the Pan condition is still satisfied.
\end{proof}

\begin{proposition}
Let $G= (N,E,w)$ be a connected weighted graph
satisfying the Star, Path conditions
and the Cycle condition
(or the Weak Second Part of the Adjacent Cycles condition).
Then Pan condition can be verified in $O(n^{5})$ time.
\end{proposition}

\begin{proof}
Let $C = \lbrace 1, e_{1}, 2, e_{2}, \ldots, m, e_{m}, 1 \rbrace$
be a simple cycle of $G$.
The Cycle condition implies
$w_{1} \leq w_{2} \leq w_{3} = \cdots = w_{m} = M = \max_{e \in E(C)} w(e)$,
after renumbering if necessary.
Let $\tilde{\gamma}$ be a shortest path linking $1$ and $3$
in $G_{N \setminus \lbrace 2 \rbrace}$
and let $\tilde{C}$ be the simple cycle formed by $e_{1}$, $e_{2}$
and $\tilde{\gamma}$.
If $\tilde{C} \not= C$
then Lemma~\ref{lemM=maxu(e)=maxu(e)=M'}
implies that $\max_{e \in E(\tilde{C})} = M$.
Lemma~\ref{lemPanConditionWithAssumptionStarPathCycleConditionsSatisfied}
implies that Pan condition for $C$
is equivalent to Pan condition for $\tilde{C}$.
Then if $w_{1} = w_{2} < M$
either the chord $\lbrace 1, 3 \rbrace$ exists
(\emph{i.e.} $\tilde{\gamma} = \lbrace 1, \lbrace 1, 3 \rbrace, 3 \rbrace$
and $\tilde{C} = \lbrace 1, e_{1}, 2, e_{2}, 3, \lbrace 1, 3 \rbrace, 1 \rbrace$)
otherwise $w_{1} = w_{2} = \sigma(E)$.
Hence to verify the Pan condition we just need to consider
any pair of adjacent edges $e_{1} = \lbrace 1, 2 \rbrace$
and $e_{2} = \lbrace 2,3 \rbrace$
with $w_{1} = w_{2}$.
Let us consider the graph $\tilde{G} = (N,  \tilde{E}, w_{|\tilde{E}})$
with $\tilde{E} = \lbrace e \in E; \; w(e) \geq w_{1} \rbrace$.
Then we look for a shortest path $\gamma$ linking $1$ to $3$
in the graph $\tilde{G}_{N \setminus \lbrace 2 \rbrace}$
using a BFS algorithm in $O(n^{2})$ time.
If $\gamma$ exists
and if $\gamma$ has at least two
edges,
we consider the simple cycle $C = \lbrace 1, e_{1}, 2, e_{2}, 3 \rbrace \cup \gamma$.
Note that
as $\gamma$ is a shortest path linking $1$ and $3$,
$\lbrace 1, 3 \rbrace$ cannot be a chord of $\tilde{C}$.
Then applying
Lemma~\ref{lemPanConditionWithAssumptionStarPathCycleConditionsSatisfied} to $\tilde{C}$
we only need  to verify that
either $w_{1} = w_{2} = w_{3}$
or $w_{1} = w_{2} = \sigma(E)$.
$\sigma(E)$ can be computed in $O(m)$ time.
We have at most
$\sum_{i \in V} C_{n}^{2} = n^{2} \left(\frac{n-1}{2}\right)$
pairs $e_{1}, e_{2}$ to take into account.
Therefore the Pan condition can be verified in $O(n^{5})$.
\end{proof}

\begin{proposition}
Let $G = (N,E,w)$ be a weighted graph satisfying
the Weak Second Part of the Adjacent Cycles
and the Pan conditions.
Then the Second part of the Adjacent Cycles condition is satisfied.
\end{proposition}

\begin{proof}
Let $e_{1}$ be a unique non-maximum weight edge common
to two simple cycles $(C, C^{'})$ in $G$
satisfying the conditions
\ref{enumPropV(C)-V(C)^{'}nonempty},
\ref{enumPropCatmost1non-maxweightchord},
\ref{enumPropCC^{'}nomaxweightchord},
\ref{enumPropCandC'HaveNoCommonChord}
of the Adjacent Cycles condition.
The Weak Second Part of the Adjacent Cycles condition implies
that there are non-maximum weight edges $e_{2} \in E(C) \setminus \lbrace e_{1} \rbrace$
and $e_{2}^{'} \in E(C^{'}) \setminus \lbrace e_{1} \rbrace$
such that $e_{1}, e_{2}, e_{2}^{'}$ are adjacent
and $w_{1} = w_{2} \geq w_{2}^{'}$
or $w_{1} = w_{2}^{'} \geq w_{2}$.
If $|E(C)| = 3$ or $|E(C^{'})| = 3$
then the second part of the Adjacent Cycles condition is satisfied.
If $|E(C)| \geq 4$ and $|E(C^{'})| \geq 4$,
let us assume $w_{1} = w_{2} > w_{2}^{'}$ (resp. $w_{1} = w_{2}^{'} > w_{2}$).
By assumption $C$ (resp. $C^{'}$) have no maximum weight chord
therefore the Pan condition applied to $C$ and $e_{2}^{'}$
(resp. $C^{'}$ and $e_{2}$)
implies $w_{1} = w_{2} = \max_{e \in E(C)} w(e)$
(resp. $w_{1} = w_{2}^{'} = \max_{e \in E(C^{'})} w(e)$),
a contradiction.
Therefore $w_{1} = w_{2} = w_{2}^{'}$
and the second part of the Adjacent Cycles condition is satisfied.
\end{proof}

\begin{theorem}
Let $G= (N,E,w)$ be an arbitrary weighted graph.
Let us consider the correspondance $\mathcal{P}_{\min}$.
We can verify inheritance of $\mathcal{F}$-convexity
from $(N,v)$ to $(N, \overline{v})$
for all $\mathcal{F}$-convex and superadditive game $(N,v)$
in $O(n^{6})$ time.
\end{theorem}

\section{Conclusion}
We think that the method we have presented
to prove that inheritance of $\mathcal{F}$-convexity
is a polynomial problem,
especially,
the use of minimum spanning trees and shortest paths
to select specific paths and cycles,
can be useful for further developments.
In particular it could be used
to obtain polynomial algorithms
to check inheritance of convexity for $\mathcal{P}_{\min}$
or of convexity or $\mathcal{F}$-convexity
for other correspondences
as for example the correspondence $\mathcal{P}_{G}$ presented in \citep{GrabischSkoda2012}
which associates to a subset its partition into relatively strongest components.
Of course it is not certain that such algorithms do exist.

\section*{Acknowledgments}
The support of the French National Research Agency,
through DynaMITE project (Reference: ANR-13-BSH1-0010-01),
is gratefully acknowledged.

\bibliographystyle{apalike}
\bibliography{biblio}


\end{document}